\documentclass{article}

\usepackage[latin1]{inputenc}
\usepackage[centertags]{amsmath}
\usepackage{amsfonts}
\usepackage{amssymb}
\usepackage{amsthm}
\usepackage{newlfont}
\usepackage{mathrsfs}
\usepackage{amsbsy}

\newcommand{\nat}{\mathbb N}
\newcommand{\real}{\mathbb R}
\newcommand{\F}{\mathcal{F}}

\newcommand{\M}{\mathcal{M}}
\newcommand{\R}{\mathcal{R}}
\newcommand{\Ri}{\mathcal{R}^\infty}
\newcommand{\D}{\mathcal{D}}
\newcommand{\Q}{\mathcal{Q}}
\newcommand{\A}{\mathcal{A}}

\newcommand{\T}{\mathbb{T}}

\newcommand{\X}{\mathcal{X}}

\newcommand{\ma}{\M(P)}
\newcommand{\me}{\M^e(P)}
\newcommand{\mal}{\M_{\text{loc}}(P)}
\newcommand{\mel}{\M_{\text{loc}}^e(P)}
\newcommand{\map}{\M(\bar P)}
\newcommand{\mep}{\M^e(\bar P)}

\newcommand{\LT}{L^{\infty}}
\newcommand{\Lt}{L^{\infty}_t}
\newcommand{\LTs}{L^{\infty}}
\newcommand{\Lts}{L^{\infty}_t}
\newcommand{\LTp}{\bar{L}^\infty}
\newcommand{\Ltp}{\bar{L}^\infty_t}
\newcommand{\LTf}{L^{\infty}(\bar\Omega, \bar\F,\bar{P})}
\newcommand{\Ltf}{L^{\infty}(\bar\Omega, \bar\F_t,\bar{P})}

\newcommand{\lk}{\left\{\,}
\newcommand{\rk}{\right\}}
\newcommand{\mk}{\;\big|\;}

\newcommand{\lp}{\left[}
\newcommand{\rp}{\right]}

\DeclareMathOperator*{\es}{ess\,sup}

\DeclareMathOperator*{\qes}{Q\text{-}ess\,sup}
\DeclareMathOperator*{\qesp}{\bar{Q}\text{-}ess\,sup}

\newcommand{\pk}{,\ldots ,}
\newcommand{\pke}{,\ldots }

\newcommand{\zte}{_{t\in\T}}
\newcommand{\ztu}{_{t\in\T\cap\nat_0}}
\newcommand{\zt}{_{t\in\T\cap\nat_0}}

\newcommand{\rt}{\rho_t}
\newcommand{\rtf}{$(\rt)\ztu$ }

\newcommand{\qf}{Q\mbox{-a.s.}}

\newcommand{\ctse}{continuous from above}

\newtheorem{theorem}{Theorem}
\newtheorem{assumption}[theorem]{Assumption}
\newtheorem{corollary}[theorem]{Corollary}
\newtheorem{definition}[theorem]{Definition}
\newtheorem{lemma}[theorem]{Lemma}
\newtheorem{proposition}[theorem]{Proposition}
\newtheorem{remark}[theorem]{Remark}
\newtheorem{example}[theorem]{Example}

\title{Risk assessment for uncertain cash flows:\\ Model ambiguity, discounting ambiguity, and the role of bubbles}
\author{Beatrice Acciaio\thanks{Department of Economy, Finance and Statistics, University of Perugia, Via A. Pascoli 20, 06123 Perugia, Italy. Email: beatrice.acciaio@stat.unipg.it. Financial support from the European Science Foundation (ESF) ``Advanced Mathematical Methods for Finance" (AMaMeF) under the exchange grant 2281, and hospitality of Vienna University of Technology are gratefully acknowledged.} \and Hans F{\"o}llmer\thanks{Humboldt-Universit\"{a}t zu Berlin, Institut f\"{u}r Mathematik, Unter den Linden 6, 10099 Berlin, Germany. Email: foellmer@math.hu-berlin.de.} \and Irina Penner\thanks{Humboldt-Universit\"{a}t zu Berlin, Institut f\"{u}r Mathematik, Unter den Linden 6, 10099 Berlin, Germany. Email: penner@math.hu-berlin.de. Supported by the DFG Research Center \textsc{Matheon} ``Mathematics for key technologies''. Financial support from the European Science Foundation (ESF) ``Advanced Mathematical Methods for Finance" (AMaMeF) under the short visit grant 2854 is gratefully acknowledged.}}
\begin{document}
\maketitle
\begin{abstract}
We study the risk assessment of uncertain cash flows in terms of dynamic convex risk measures for processes as introduced in Cheridito, Delbaen, and Kupper \cite{cdk6}.  These risk measures take into account not only the amounts but also the timing of a cash flow. We discuss their robust representation in terms of suitably penalized probability measures on the optional $\sigma$-field. This yields an explicit analysis both of model and discounting ambiguity. We focus on supermartingale criteria for different notions of time consistency. In particular we show how ``bubbles'' may appear in the dynamic penalization, and how they cause a breakdown of asymptotic safety of the risk assessment procedure. 
\end{abstract}

\section{Introduction}
The classical assessment of an uncertain cash flow takes the sum of the discounted future payments and computes its expectation with respect to a given probability measure. Both the probabilistic model and the discounting factors are assumed to be known. In reality, however, one is usually confronted both with model uncertainty and with uncertainty about the time value of money. The purpose of this paper is to deal with this problem by using concepts and methods from the theory of convex risk measures.

In a situation where financial positions are described by random variables on some probability space, a convex risk measure can usually be represented as the worst expected loss over a class of suitably penalized probabilistic models; see Artzner, Delbaen, Eber, and Heath~\cite{adeh97,adeh99}, Delbaen~\cite{d0, d2} for the coherent case, and F\"ollmer and Schied~\cite{fs2,fs4}, Frittelli and Rosazza Gianin~\cite{fr2} for the general convex case. This can be seen as a robust method which deals explicitly with the problem of model uncertainty. In the dynamical setting of a filtered probability space, the risk assessment at a given time should depend on the available information. This is specified by a dynamic risk measure, i.e., by a sequence $(\rt)$ of conditional convex risk measures adapted to the filtration. On the level of random variables, and under an additional requirement of time consistency, the structure of such dynamic risk measures is now well understood; cf., e.g, \cite{adehk7,rse5,dt5,d6,Weber,ks5,bn6,fp6,ck6,tu8, ipen7, dpr10, ap9}, and references therein.

There is also a growing literature on dynamic risk measures applied to cash flows that are described as adapted stochastic processes on the given filtered probability space; cf., e.g., \cite{adehk7, rie4, cdk4, cdk5, cdk6, ck6, fs6, jr8}.
In this context, not only the amount of a payment matters, but also its timing. In particular, the risk is reduced by having positive payments earlier and negative ones later. This is expressed by the property of cash subadditivity, which was introduced by El Karoui and Ravanelli~\cite{er08} in the context of risk measures for random variables in order to account for discounting ambiguity. Convex risk measures for processes have that property, and so they provide a natural framework to capture both model uncertainty and uncertainty about the time value of money.

In this paper we study dynamic convex risk measures for bounded adapted processes, as introduced in \cite{cdk6}. Any such process can be viewed as a bounded measurable function on the product space $\bar{\Omega}=\Omega\times\T$ endowed with the optional $\sigma$-field. It is thus natural to use results from the theory of risk measures for random variables and to apply them on product space. This idea already appears in \cite{adehk7} in a static setting. Here we use it for dynamic risk measures, and we take a more probabilistic approach. This involves a careful study of absolutely continuous probability measures $\bar Q$ on the optional $\sigma$-field. In particular, we derive a decomposition $\bar Q=Q\otimes D$, where $Q$ is a locally absolutely continuous probability measure on the original space, and $D$ is a predictable discounting process. The probabilistic approach has two advantages. In the first place, it allows us to make explicit the joint role of model uncertainty, as expressed by the measures $Q$, and of discounting uncertainty, as described by the discounting processes $D$, in the robust representation of conditional risk measures. Moreover, it is crucial for our analysis of the supermartingale aspects of time consistency.

A key issue in the dynamical framework is time consistency of the risk assessment; see \cite{adehk7,d6,dt5,ks5,cdk6,bn6,fp6,ck6,dpr10}, and references therein. We characterize time consistency by supermartingale properties of the discounted penalty and risk processes, in analogy to various results for random variables from \cite{adehk7,d6, bn6, fp6, ipen7, bn8}. These characterizations allow us to apply martingale arguments to prove maximal inequalities and convergence results for the risk assessment procedure. In particular, we show that the appearance of a martingale component in the Riesz  decomposition of the discounted penalty process  amounts to a breakdown of asymptotic safety. Such a martingale can be seen as a ``bubble", which appears on the top of the ``fundamental" penalization and thus causes an excessive neglect of the model under consideration.

The paper is organized as follows. In Section~\ref{cond} we clarify the probabilistic structure of conditional convex risk measures for processes. To this end, we introduce the appropriate product space in Subsection~\ref{prodspace} and state a decomposition theorem for measures on the optional $\sigma$-field; its proof is given in Appendix~\ref{proof}. In Subsection~\ref{correspsec} risk measures for processes are identified with risk measures for random variables on the product space. This allows us to obtain a robust representation of risk measure for processes in Subsection~\ref{robsec} which involves both model ambiguity and discounting ambiguity.
Section~\ref{tc} characterizes time consistency of dynamic risk measures, with special emphasis on the corresponding supermartingale properties. We first focus on the strong notion of time consistency. In Subsection~\ref{stc} we state several equivalent criteria. They are used in Subsection~\ref{rieszsec} to derive the Doob and the Riesz decomposition of the penalty processes. In Subsection~\ref{asympsec} we discuss asymptotic properties such as asymptotic safety and asymptotic precision, and we relate them to the appearance of ``bubbles" in the Riesz decomposition. Subsection~\ref{maxsec} states a maximal inequality for the excess of the capital requirement over the penalized expected loss computed for a specific model.
The coherent case is discussed in Subsection~\ref{cohsec}, and some weaker notions of time consistency are introduced and characterized in Subsection~\ref{weaksec}. In Section~\ref{cashadsec} we discuss cash subadditivity of risk measures for processes, and we characterize their calibration with respect to some num{\'e}raire. If a time consistent dynamic risk measure is calibrated to a term structure specified by the prices of zero coupon bonds, then discounting ambiguity is completely resolved, and we are only left with model ambiguity. In Section~\ref{examplesec} our analysis is illustrated by some examples, including entropic risk measures and variants of Average Value at Risk for processes.

\section{Preliminaries}\label{prem}
In this paper we consider a discrete-time setting with time horizon $T\in\nat\cup\{\infty\}$. We denote by $\T$ the set of time points, i.e., $\T:=\{0\pk T\}$ if $T<\infty$, and in case $T=\infty$ we distinguish between the two  cases
$\T:=\nat_0$ and $\T:=\nat_0\cup\{\infty\}$.
We use the notation $\T_t:=\{s\in\T\,|\,s\ge t\}$ for $t\in\T$.

We fix a filtered probability space 
$(\Omega, \F, (\F_t)\zt, P)$, with $\F_0=\{\emptyset, \Omega\}$,  and $\displaystyle{\F_\infty:=\sigma(\cup_{t\in\nat_0}\F_t)}$ for $T=\infty$. For $t\in\T$, we use the notation 
\[
L^\infty_t:=L^{\infty}(\Omega, \F_t,P),\qquad L_{t,+}^{\infty}:=\{X\in L^{\infty}_t\mk X\geq 0\},
\]
and  $L^{\infty}:=L^{\infty}(\Omega, \F_T,P)$. All equalities and inequalities between random variables and between sets are understood to hold $P$-almost surely, unless stated otherwise. 

We denote by $\ma$ (resp.\ by $\mal$) the set of all probability measures $Q$ on $(\Omega, \F)$ which are absolutely continuous with respect to $P$ (resp.\ locally absolutely continuous with respect to $P$ in the sense that $Q\ll P$ on $\F_t$ for each $t\in\T\cap\nat_0$), and by $\me$ (resp.\ by $\mel$) the set of all probability measures on $(\Omega, \F)$ which are equivalent (resp.\ locally equivalent) to $P$. Note  that $\ma$ coincides with $\mal$ if $T<\infty$.

Let $\mathcal{R}^{\infty}$ denote the space of adapted stochastic processes $X=(X_t)\zte$ on $(\Omega, \F, (\F_t)\zte, P)$ such that
\[
\|X\|_{\infty}:=\inf \left\{x\in\real \mk \sup\zte|X_t|\leq x\right\}<\infty.
\]
For $T=\infty$ we also consider the subspace
\[
 \X^\infty:=\lk X\in\R^\infty\mk \exists\,X_\infty=\lim_{t\to\infty}X_t\;P\text{-a.s.}\rk.
\]
For $0\leq t\leq s\leq T$, we define the projection $\pi_{t,s}:\mathcal{R}^{\infty}\to \mathcal{R}^{\infty}$ as 
\[
\pi_{t,s}(X)_r=1_{\{t\leq r\}}X_{r\wedge s},\quad r\in\T,
\]
and use the notation $\mathcal{R}_{t,s}^{\infty}:=\pi_{t,s}(\mathcal{R}^{\infty})$ and $\mathcal{R}_t^{\infty}:=\pi_{t,T}(\mathcal{R}^{\infty})$. The spaces $\X_{t,s}^\infty$ and $\X_t^\infty$ are defined accordingly.

We interpret a process $X\in\R^\infty$  as a \emph{cumulated cash flow}, as explained in Remark~\ref{rem:cashflow} and in Example~\ref{cashflow}, or as a value process, which might model the evolution of some financial value such as the market value of a firm's equity or of an investment portfolio. 
\begin{remark}\label{rem:cashflow}
 An adapted cash flow $C=(C_t)_{t\in\T\cap\nat_0}$ yielding an uncertain amount $C_t\in\Lt$ at time $t$ induces a cumulated cash flow $X=(X_t)_{t\in\T\cap\nat_0}$ with
\[
 X_t=\sum_{s=0}^tC_t.
\]
If $T<\infty$, or if $T=\infty$ and $\sum_{t\in\T\cap\nat_0}\|C_t\|_\infty<\infty$, the process $X$ belongs to $\Ri$, and even to $\X^\infty$, with $X_\infty:=\sum_{s=0}^\infty C_t$.  Conversely, each process $X\in\Ri$ induces an adapted cash flow
\[
 C_t:=\Delta X_t:=X_t-X_{t-1},\qquad t\in\T\cap\nat_0,
\]
where we use the convention $X_{-1}:=0$. 
\end{remark}

\begin{example}\label{cashflow}
Assume that there is a money market account $(B_t)_{t\in\T\cap\nat_0}$ of the form
\[
B_t=\prod_{s=1}^t(1+r_s) 
\]
with some adapted (or even predictable) process $(r_t)\zt$ of nonnegative short rates. For a given (undiscounted) adapted cash flow $(\tilde C_t)\zt\in\Ri$ consider the discounted cash flow $C=(C_t)_{t\in\T\cap\nat_0}$ defined by $C_t=B_t^{-1}\tilde C_t$. If $T<\infty$, or if $T=\infty$ and the short rates are bounded away from zero by some constant $\delta>0$, then the discounted cash flow $C$ belongs to $\Ri$, and for $T=\infty$ even to $\X^\infty$, since 
\[
 \sum_{t=0}^\infty\|C_t\|_\infty\le\frac{1}{\delta}\|\tilde C\|_\infty<\infty.
\]

\end{example}

\section{Conditional risk measures}\label{cond}
At each time the risk of a future cumulative cash flow will be assessed by a conditional risk measure based on the information available at that time. The following definition  was introduced in \cite{cdk6}.

\begin{definition}\label{defrmp}
A map $\rt\,:\,\mathcal{R}_t^{\infty}\,\rightarrow\,\Lt$ for $t\in\T\cap\nat_0$ is called a \emph{conditional convex risk measure (for processes)} if it satisfies the following properties for all $X,Y\in\mathcal{R}_t^{\infty}$
\begin{itemize}
\item
Conditional cash invariance: for all $m\in\Lt$,
\[\rho_t(X+m1_{\T_t})=\rho_t(X)-m;\]
\item
Monotonicity: $\rt(X)\ge\rt(Y)$ if $X\le Y$ componentwise;
\item
Conditional convexity: for all $\lambda\in\Lt$ with $0\le \lambda\le 1$,
\[
\rt(\lambda X+(1-\lambda)Y)\le\lambda\rt(X)+(1-\lambda)\rt(Y);
\]
\item
{Normalization}: $\rt(0)=0$.
\end{itemize}
A conditional convex risk measure is called a \emph{conditional coherent risk measure (for processes)} if it has in addition 
the following property:
\begin{itemize}
\item
{Conditional positive homogeneity}: for all $\lambda\in\Lt$ with $\lambda\ge0$,
\[
\rt(\lambda X)=\lambda\rt(X).
\] 
\end{itemize}
A sequence $(\rt)\ztu$ is called a \emph{dynamic convex risk measure (for processes)} if, for each $t$, $\rt\,:\,\mathcal{R}_t^{\infty}\,\rightarrow\,\Lt$ is a conditional convex risk measure (for processes).
\end{definition}

Definition~\ref{defrmp} is analogous to the definition of risk measures for random variables; cf.\ Definition \ref{defrm}. Note, however, that conditional cash invariance in the context of processes takes  into account the timing of the cash payment; the consequences will be discussed in more detail in Section~\ref{cashadsec}.

\subsection{Optional filtration and predictable discounting}\label{prodspace}
It was already noted in Artzner et al.\ \cite{adehk7} that static risk measures for processes can be viewed as risk measures for random variables on an appropriate product space. In this section we extend this idea to the dynamic setting, and we focus on the probabilistic structure of the resulting robust representation in terms of probability measures on the optional $\sigma$-field. 

Consider the product space $(\bar\Omega, \bar\F,\bar{P})$ defined by
\[
\bar\Omega=\Omega\times\T,\quad \bar\F=\sigma(\{A_t\times\{t\}\mk A_t\in\F_t,\, t\in\T),\quad\bar{P}=P\otimes\mu,
\]
where $\mu=(\mu_t)\zte$ is some adapted reference process such that $\sum_{t\in\T}\mu_t=1$ and $\mu_t>0$ $\forall t\in\T$, and where
\[
E_{P\otimes\mu}[X]:=E_P\lp \sum_{t\in\T}X_t\mu_t\rp
\]
for any bounded measurable function $X$ on $(\bar\Omega, \bar\F)$.

Note that $\bar{\F}$ coincides with the \emph{optional} $\sigma$-field generated by all adapted processes.  
Every adapted process can be identified with a random variable on $(\bar{\Omega},\bar{\F}, \bar P)$, and in particular we have
\[\mathcal{R}^{\infty}=\LTp:=\LTf.\]
We also introduce the \emph{optional filtration} $(\bar\F_t)\zte$ on $(\bar\Omega,\bar\F)$ given by
\[
\bar\F_t=\sigma\left(\{A_j\times\{j\}, A_t\times\T_t\mk  A_j\in\F_j,\,j<t,\, A_t\in\F_t\}\right),\quad t\in\T.
\]
A random variable $X=(X_s)_{s\in\T}$ on $(\bar{\Omega},\bar{\F}, \bar P)$ is $\bar\F_t$-measurable if and only if $X_s$ is $\F_s$-measurable for all $s=0\pk t$ and $X_s=X_t$ $\forall s>t$. In particular, 
\[
 \mathcal{R}^{\infty}_{0,t}=\Ltp:=\Ltf.
\] 
The set $\mathcal{R}^{\infty}_{0,0}$ of all constant processes will be identified with $\real$.

For $T=\infty$ we will use the Lebesgue decomposition of a measure $Q\in\mal$ with respect to $P$. Let $M=(M_t)_{t\in\nat_0}$ denote the density process of $Q$ with respect to $P$. The limit  $M_\infty:=\lim_{t\to\infty}M_t$ exists $P$-a.s., since $M$ is a nonnegative $P$-martingale. By \cite[Theorem VII.6.1]{shi96} $M_\infty$ exists also $Q$-a.s., and $Q$ admits the Lebesgue decomposition 
\begin{equation}\label{leb}
Q[A]=E_P[I_AM_\infty]+Q[A\cap\{M_\infty=\infty\}],\quad A\in\F_\infty 
\end{equation}
into the absolutely continuous and the singular part with respect to $P$ on $(\Omega,\F_\infty)$. 

For a measure $Q\in\mal$ we introduce the set $\Gamma(Q)$ of \emph{optional random measures} $\gamma=(\gamma_t)\zte$ on $\T$ which are normalized with respect to $Q$. More precisely, $\gamma\in\Gamma(Q)$ is a nonnegative adapted process, such that 
\[
 \sum_{t\in\T}\gamma_t=1\quad Q\text{-a.s.},
\]
with the additional property that  
\[
\gamma_\infty=0\quad Q\text{-a.s.\ on}\quad \{M_\infty=\infty\},\quad\text{if}\quad \T=\nat_0\cup\{\infty\}.
\]

We also consider the following set $\D(Q)$ of \emph{predictable discounting processes}: $D=(D_t)_{t\in\T}\in\D(Q)$ is a predictable non-increasing process with $D_0=1$, and $D_\infty=\lim_{t\to\infty}D_t$ $Q$-a.s.\ for $T=\infty$, where
\[
D_\infty=0\quad\qf\quad\text{for}\quad \T=\nat_0,
\]
and
\[
D_\infty=0\quad\qf\quad\text{on}\quad\{M_\infty=\infty\}\quad\text{for}\quad \T=\nat_0\cup\{\infty\}.
\]
For $T<\infty$ we define $D_{T+1}:=0$.

\begin{lemma}\label{intbypartslem}
For any probability measure $Q\in\mal$, the set $\Gamma(Q)$ can be identified with $\D(Q)$.
More precisely, to each $\gamma$ in $\Gamma(Q)$ we can associate a process $D\in\D(Q)$  given by
\begin{equation}\label{gammatodisc}
D_t:=1-\sum_{s=0}^{t-1}\gamma_s,\quad t\in\T\cap\nat_0,\quad\text{and}\quad D_\infty:=\gamma_\infty\quad\text{for}\quad \T=\nat_0\cup\{\infty\}.
\end{equation}
In particular we have
\begin{equation}\label{gammatodisc2}
D_t=\sum_{s\in\T_t}\gamma_s\quad \qf\quad\forall t\in\T.
\end{equation}

Conversely, every process $D\in\D(Q)$ defines an optional random measure $\gamma\in\Gamma(Q)$ via
\begin{equation}\label{disctogamma}
\gamma_t:=D_t-D_{t+1},\quad t\in\T\cap\nat_0,\quad\text{and}\quad \gamma_\infty:=D_\infty\quad\text{for}\quad \T=\nat_0\cup\{\infty\}.
\end{equation}

Moreover, for any pair $\gamma\in\Gamma(Q)$ and $D\in\D(Q)$ related to each other via \eqref{gammatodisc2} and \eqref{disctogamma}, the ``integration by parts'' formula
\begin{equation}\label{intbyparts}
\sum_{s\in\T_t}\gamma_sX_s=\sum_{s=t}^T D_s(X_s-X_{s-1})\quad \qf,\quad t\in\T,
\end{equation}
holds for any $X\in\R_t^{\infty}$ if $T<\infty$ or if $\T=\nat_0$, and for $X\in\X^\infty_t$ if $\T=\nat_0\cup\{\infty\}$.
\end{lemma}
\begin{proof}
It is obvious that the process $D$ defined by \eqref{gammatodisc} belongs to $\D(Q)$ and satisfies \eqref{gammatodisc2}, and that $\gamma$ defined by \eqref{disctogamma} belongs to $\Gamma(Q)$. To prove \eqref{intbyparts}, note that
\begin{equation}\label{ibp}
 \sum_{s=0}^t\gamma_sX_s=\sum_{s=0}^tD_s(X_s-X_{s-1})-D_{t+1}X_t
\end{equation}
for all $t\in\T\cap\nat_0$.
Thus \eqref{intbyparts} is obvious for $T<\infty$, and it also holds if $\T=\nat_0$ for all $X\in\Ri_t$, since $X$ is bounded and $D_t\searrow0$ $Q$-a.s.. For $\T=\nat_0\cup\{\infty\}$ and for any $X\in\X_t^\infty$, the limit $D_\infty X_\infty=\lim_{t\to\infty}D_{t+1}X_t$ exists $Q$-a.s., since $D_t\searrow0$ $Q$-a.s. on the singular part of $Q$ with respect to $P$, and so \eqref{intbyparts} follows again from \eqref{ibp}.
\end{proof}

From now on we use the following assumption.

\begin{assumption}\label{assfil}
In the case $T=\infty$, we assume that for each $t\in\T\cap\nat_0$ the $\sigma$-field $\F_t$ is $\sigma$-isomorphic to the Borel $\sigma$-field on some complete separable metric space, and that $\cap_nA_n\neq\emptyset$ for any decreasing sequence $(A_n)_{n\in\nat_0}$ such that $A_n$ is an atom of $\F_{n}$.
\end{assumption}
We denote by $\map$ the set of all probability measures on $(\bar\Omega, \bar\F)$ which are absolutely continuous with respect to $\bar{P}$. The next theorem shows that each probability measure $\bar{Q}$ in $\map$ admits a decomposition $\bar{Q}(d\omega,dt)=Q(dw)\otimes \gamma(w,dt)$ for some probability measure $Q$ on $(\Omega,\F_T)$ and some optional random measure $\gamma$ on $\T$ such that $Q\in\mal$ and $\gamma\in\Gamma(Q)$.

\begin{theorem}\label{propq}
For any probability measure $\bar Q\in\map$ there exist a probability measure $Q\in\mal$ and an optional random measure $\gamma\in\Gamma(Q)$ (resp.\ a predictable discounting factor $D\in\D(Q)$) 
 such that
\begin{align}
E_{\bar{Q}}[X]&=E_Q\lp\sum_{t\in\T}\gamma_tX_t\rp\label{expqg}\\
&=E_Q\lp\sum_{t=0}^TD_t(X_t-X_{t-1})\rp, \label{expqd}
\end{align}
where \eqref{expqg} holds for all $X\in\Ri$, whereas \eqref{expqd} holds for all $X\in\Ri$ if $T<\infty$ or if $\T=\nat_0$,  and only for $X\in\X^\infty$ if $\T=\nat_0\cup\{\infty\}$.

Conversely, any $Q\in\mal$ and any $\gamma\in\Gamma(Q)$ (resp.\ any $D\in\D(Q)$) define a probability measure $\bar Q\in\map$ such that \eqref{expqg} and \eqref{expqd} hold.

We write
\[
\bar{Q}=Q\otimes\gamma= Q\otimes D
\]
to denote the decomposition of $Q$ in the sense of \eqref{expqg} and \eqref{expqd}.
\end{theorem}
The proof is postponed to Appendix~\ref{proof}. 

\begin{remark}
A continuous time analogue to Theorem~\ref{propq} appears independently in Kardaras~\cite[Theorem 2.1]{kard10}. While we make use of the It\^o-Watanabe decomposition (in discrete time, cf. Proposition~\ref{lemiw}) and of a measure theoretic extension, \cite[Theorem 2.1]{kard10} gives a direct construction of a discounting process and a local martingale, without relating the latter to a probability measure $Q$ in the general case.
\end{remark}

\subsection{Conditional risk measures viewed on the optional filtration}\label{correspsec}
In the previous section we have identified processes in $\R^{\infty}$ with random variables in $\LTp$. This induces a one-to-one correspondence between conditional risk measures for processes and conditional risk measures for random variables on the optional $\sigma$-field:
\begin{proposition}\label{corresp}
Any conditional convex risk measure for processes $\rt\,:\,\mathcal{R}_t^{\infty}\,\rightarrow\,\Lt$ for $t\in\T\cap\nat_0$ defines a conditional convex risk measure on random variables $\bar\rho_t\,:\,\LTp\,\rightarrow\,\Ltp$ via 
\begin{equation}\label{rrp}
\bar\rho_t(X)=-X_01_{\{0\}}-\ldots-X_{t-1}1_{\{t-1\}}+\rt(X)1_{\T_t},\quad X\in \R^\infty,
\end{equation}
where we use the notation  
\[
\rt(X):=\rt\circ\pi_{t,T}(X)\quad\text{for}\quad  X\in\mathcal{R}^{\infty}.
\]
Conversely, any conditional convex risk measure on random variables $\bar\rho_t\,:\,\LTp\,\rightarrow\,\Ltp$ is of the form \eqref{rrp} with some  conditional convex risk measure on processes $\rt\,:\,\mathcal{R}_t^{\infty}\,\rightarrow\,\Lt$.
\end{proposition}
\begin{proof}
Clearly, $\bar{\rt}$ defined via \eqref{rrp} is a conditional convex risk measure in the sense of Definition~\ref{defrm}. To see, e.g., conditional cash invariance, let $m\in \Ltp$, i.e. $m=(m_0\pk m_{t-1},m_t,m_t\pke)$ with $m_i\in L^\infty_i$ for $i=0\pk t$. Then
\[
\bar\rho_t(X+m)=(-X_0-m_0\pk -X_{t-1}-m_{t-1},\rt(X+m),\rt(X+m)\pke)=\bar\rho_t(X)-m
\]
by conditional cash invariance of $\rt$. \\ To prove the converse implication, let $\bar\rho_t\,:\,\LTp\,\rightarrow\,\Ltp$ be a conditional convex risk measure for random variables. 
Since $A_t:=\Omega\times\{0\pk t-1\}\in\bar\F_t$, the local property (cf., e.g., \cite[Proposition 2]{dt5}), conditional cash invariance and normalization of $\bar\rho_t$ imply
\begin{equation*}
\bar\rt(X)=I_{A_t}\bar\rt(I_{A_t}X)+I_{A_t^c}\bar\rt(I_{A_t^c}X)=-X_0I_{\{0\}}-\cdots-X_{t-1}I_{\{t-1\}}+\bar\rt(XI_{\T_t})I_{\T_t}.
\end{equation*}
Finally, it is easy to see that $\rt\,:\,\mathcal{R}_t^{\infty}\,\rightarrow\,\Lt$ defined by $\rt(X):=(\bar\rho_t(X))_t$ is a conditional convex risk measure for processes in the sense of Definition~\ref{defrmp}.
\end{proof}

Let $\rt\,:\,\mathcal{R}_t^{\infty}\,\rightarrow\,\Lt$ be a conditional convex risk measure for processes, and consider the corresponding acceptance set 
\begin{equation*}
\mathcal{A}_t=\{X\in\mathcal{R}_t^{\infty}\mk \rt(X)\leq 0\}.
\end{equation*}
Then the acceptance set of $\bar\rt$ related to $\rt$ via \eqref{rrp} is given by
\begin{eqnarray}\label{accsetp}
\bar{\A}_t&=&\lk X\in\LTp\mk \bar\rho_t(X)\le0\;\bar{P}\text{-a.s.}\rk \nonumber\\
&=&\lk X\in\R^\infty\mk X_s\geq 0\, \forall s=0\pk t-1,\; \rt(X)\leq 0\; P\text{-a.s.}\rk \nonumber\\
&=&\A_t+L_{0,+}^{\infty}\times\ldots\times L_{t-1,+}^{\infty}\times\{0\}\times\ldots.
\end{eqnarray}

For each $\bar{Q}\in\map$, the minimal penalty function of $\bar\rt$ is given by
\begin{equation*}%\label{minpenp}
\bar\alpha_t(\bar{Q})=\qesp_{X\in\bar{\A}_t}E_{\bar{Q}}[-X\,|\,\bar\F_t\,].
\end{equation*}
Due to \eqref{accsetp} and Corollary~\ref{cor:condexp}, this takes the form
\begin{equation}\label{corpen}
\bar\alpha_t(\bar{Q})=\alpha_t(\bar{Q})1_{\T_t},
\end{equation}
where $\alpha_t(\bar{Q})$ denotes the minimal penalty function of $\rt$ and is given by  
\begin{align}\label{accsetpp}
\alpha_t(Q\otimes\gamma)=\alpha_t(Q\otimes D)&=\qes_{X\in\A_t}E_Q\lp-\sum_{s\in\T_t}\frac{\gamma_s}{D_t}X_s\mk\F_t\rp\\&= \qes_{X\in\R^\infty}\left(E_Q\lp-\sum_{s\in\T_t}\frac{\gamma_s}{D_t}X_s\mk\F_t\rp-\rt(X)\right)\nonumber.
\end{align}
Here $Q\otimes D=Q\otimes \gamma$ denotes the decomposition of the measure $\bar Q$ in the sense of Theorem~\ref{propq}.
Note that $\alpha_t(Q\otimes\gamma)$ is well defined $Q$-a.s.\ on $\{D_t>0\}$; cf.\ Corollary~\ref{cor:condexp}.

\subsection{Robust representations}\label{robsec}
In this section we derive a robust representation of a conditional convex risk measure for processes which expresses explicitly the combined role of model ambiguity and discounting ambiguity. Our proof will consist in combining the robust representation of risk measures for random variables as stated in \cite{dt5}, \cite{bn4}, \cite{bn8}, \cite{ks5}, \cite{fp6}, and \cite{ap9}, with our Decomposition Theorem~\ref{propq} for measures on the optional $\sigma$-field. 

The following continuity property was introduced in \cite[Definition 3.15]{cdk6}.

\begin{definition}
A conditional convex risk measure $\rt\,:\,\mathcal{R}_t^{\infty}\,\rightarrow\,\Lt$ for processes is called continuous from above if
\[
\rt(X^n)\nearrow\rt(X)\quad P\text{-a.s with}\;\,n\to\infty
\]
for any decreasing sequence $(X^n)_n\subseteq\R^{\infty}$ and $X\in\R^{\infty}$ such that $X^n_s\searrow X_s\;\,P$-a.s  for all $s\in\T_t$.
\end{definition}

\begin{theorem}\label{robdarpr}
A conditional convex risk measure for processes $\rt$ is continuous from above if and only if it admits the following robust representation:
\begin{equation}\label{newreprinf}
\rt(X)=\es_{Q\in\Q^{\text{loc}}_t}\es_{\gamma\in\Gamma_t(Q)}\left(E_Q\lp-\sum_{s\in\T_t}\gamma_sX_s\mk\F_t\rp-\alpha_t(Q\otimes\gamma)\right),   \qquad X\in\R_t^\infty
\end{equation}
where $\alpha_t$ is defined in \eqref{accsetpp},
\begin{equation*}%\label{qtloc}
\Q^{\text{loc}}_t:=\lk Q\in\mal\mk Q=P\;\,\text{on}\;\,\F_t\rk,
\end{equation*}
and
\begin{equation*}%\label{gtq}
\Gamma_t(Q):=\lk \gamma\in\Gamma(Q)\mk \gamma_s=0\,\,\forall\;\, s<t\rk.
\end{equation*}

\end{theorem}

\begin{proof}
It is easy to check that $\rt$ is continuous from above if and only if the conditional risk measure $\bar\rho_t$ defined in \eqref{rrp} is continuous from above. By Theorem~\ref{robdar}, continuity from above of $\bar\rt$ is equivalent to the robust representation
\begin{equation*}%\label{reprp}
\bar\rho_t(X)=\es_{\bar Q\in\bar\Q_t}\left(E_Q\lp-X\mk\bar\F_t\rp-\bar\alpha_t(\bar Q)\right),
\end{equation*}
where 
\begin{equation}\label{bqt}
\bar\Q_t:=\lk \bar{Q}\in\map\mk \bar{Q}=\bar{P}\;\,\text{on}\;\,{\bar\F_t}\rk.
\end{equation}
Using Corollary~\ref{cor:condexp}, this takes the form
\begin{equation}\label{reprp}
\bar\rho_t(X)=-X_01_{\{0\}}-\ldots-X_{t-1}1_{\{t-1\}}+\es_{Q\otimes\gamma\in\bar\Q_t}\left(E_Q\lp-\sum_{s\in\T_t}\frac{\gamma_s}{D_t}X_s\mk\F_t\rp-\alpha_t(Q\otimes\gamma)\right)1_{\T_t},
\end{equation}
where $D$ is related to $\gamma$ via \eqref{gammatodisc}. Lemma~\ref{rmkcon} implies that $Q\otimes\gamma\in\bar\Q_t$ if and only if $Q\in\Q_t^{\text{loc}}$, and $\gamma_s=\mu_s$ for $s=0\pk t-1$; in particular $D_t=\sum_{s\in\T_t}\mu_s>0$. For each $Q\in\Q_t^{\text{loc}}$ we can identify the set $\{(\frac{\gamma_s}{D_t})_{s\in\T_t}\mk Q\otimes\gamma\in\bar\Q_t\}$ with $\Gamma_t(Q)$, and so the  representation \eqref{newreprinf} follows from \eqref{reprp} due to \eqref{rrp}.
\end{proof}

Using the integration by parts formula \eqref{intbyparts} we can rewrite \eqref{newreprinf} as follows.

\begin{corollary}\label{robdardisc}
In terms of discounting factors, the representation \eqref{newreprinf} takes the following form for $X\in\Ri_t$ if $T<\infty$ or if $\T=\nat_0$, and for $X\in\X_t^\infty$ if $\T=\nat_0\cup\{\infty\}$:
\begin{equation}\label{newreprinfdisc}
\rt(X)=\es_{Q\in\Q^{\text{loc}}_t}\es_{D\in\D_t(Q)}\left(E_Q\lp-\sum_{s=t}^TD_s\Delta X_s\mk\F_t\rp-\alpha_t(Q\otimes D)\right), 
\end{equation}
where
\begin{equation*}%\label{gtsdiscq}
\D_t(Q)=\lk D\in\D(Q)\mk D_s=1\,\; \forall\;\, s\le t \rk.
\end{equation*}

\end{corollary}

\begin{remark}\label{rem:cdk}
In \cite{cdk6} Cheridito, Delbaen, and Kupper consider the cases $T<\infty$ and $\T=\nat_0$. They work on the space $\R^\infty$ equipped with the dual space
\[
 \A^1:=\lk a=(a_t)_{t\in\T}\mk a \,\text{adapted},\, E_P\left[\sum_{t\in\T}|a_t-a_{t-1}|\right]<\infty\rk,
\]
where $a_{-1}:=0$. The robust representation of conditional convex risk measures in \cite{cdk6} is formulated in terms of the set
\[
 \D_{0,T}:=\lk a\in\A^1\mk a_t\ge a_{t-1}\;\text{for all}\;t\in\T,\, E_P\left[\sum_{t\in\T}(a_t-a_{t-1})\right]=1\rk;
\]
cf.\ \cite[Theorem 3.16]{cdk6}. Note that $\D_{0,T}$ can be identified with the set $\map$. Indeed, every $a\in\D_{0,T}$ defines a density $\bar{Z}$ of $\bar Q\in\map$ via
\[
 Z_t\mu_t=a_t-a_{t-1},\qquad t\in\T,
\]
and vice versa. By emphasizing $\map$ rather than $\D_{0,T}$ we take a more probabilistic approach. In particular, we exploit the decomposition $\bar Q=Q\otimes\gamma=Q\otimes D$ of probability measures in $\map$. This has two advantages. In the first place it allows us to make explicit the joint role of model uncertainty, as expressed by the measures $Q\in\mal$, and of discounting uncertainty, as described by the discounting processes $D\in\D(Q)$. Moreover, the probabilistic approach allows us to discuss the case $T=\infty$ in terms of a measure theoretic extension problem, and it will be crucial for our analysis of the supermartingale aspects of time consistency.

As a special case, our representation \eqref{newreprinfdisc} applied for $T=1$ at $t=0$ to the process $(0, X_T)$ with $X_T\in L^\infty$, yields the representation (4.5) in \cite[Corollary 4.4]{er08} in the static context of cash subadditive risk measures for random variables; cf.\ also Remark~\ref{redekr}.
\end{remark}

In the same way as in Theorem~\ref{robdarpr}, the robust representations \eqref{rd:fp}, \eqref{rd:ap} and \cite[Lemma 3.5]{fp6} for conditional convex risk measures for random variables translate into representations in our context which use a smaller set of measures:
\begin{corollary}\label{cororep}
A conditional convex risk measure on processes $\rt$ is continuous from above if and only if any of the following representations hold:
\begin{enumerate}
\item 
$\rt$ is of the form \eqref{newreprinf}, where the essential supremum is taken over the set
\[
\lk Q\otimes\gamma\mk Q\in\Q^{\text{loc}}_t,\; \gamma\in\Gamma_t(Q),\; E_Q\Big[ \big(\sum_{s\in\T_t}\mu_s\big)\alpha_t(Q\otimes\gamma)\Big]<\infty\rk.
\]
\item 
for all $\bar{Q}=Q\otimes D\in\map$ and $X\in\R_t^\infty$ we have
\[
\rt(X)=\qes_{R\otimes\xi\in\bar\Q_t(\bar{Q})}\left(\frac{1}{D_t}E_R\lp-\sum_{s\in\T_t}\xi_sX_s\mk\F_t\rp-\alpha_t(R\otimes\xi)\right)\quad \qf\;\;\textrm{on}\;\, \{D_t>0\}
\]
where
\begin{equation*}%\label{qtb}
\bar\Q_t(\bar{Q}):=\lk \bar{R}\in\map\mk \bar{R}=\bar{Q}|_{\bar\F_t}\rk.
\end{equation*}
\end{enumerate}
Moreover, if there exists a probability measure $\bar P^*\approx \bar P$ on $(\bar\Omega, \bar\F)$ such that $\alpha_t(\bar P^*)<\infty$, then continuity from above is also equivalent to a representation of the form \eqref{newreprinf} as an essential supremum over the set
\[
\{Q\otimes \gamma\mk Q\in\mel,\; \gamma\in\Gamma^e(Q)\}, 
\]
where
\[
\Gamma^e(Q):= \lk \gamma\in\Gamma(Q)\mk \gamma_t>0\;P\text{-a.s.\ for all}\;t\in\T \rk.
\]

\end{corollary}

\section{Supermartingale criteria for time consistency}\label{tc}
In this section we discuss different notions of time consistency and derive corresponding criteria in terms of supermartingales.
\subsection{Strong time consistency and its characterization}\label{stc}
A strong notion of time consistency for risk measures for processes was introduced and characterized in \cite{cdk6} and \cite{ck6}. Here we adopt the definition from \cite{cdk6}, 
cf.\ \cite[Definition 4.2, Proposition 4.4, Proposition 4.5]{cdk6}. 
\begin{definition}
A dynamic convex risk measure for processes $(\rt)_{t\in\T\cap\nat_0}$ on $\R^{\infty}$ is called \emph{(strongly) time consistent} if for all $t$ in $\T$ such that $t<T$ and for all $X,Y\in\R^{\infty}$
\begin{equation}\label{tcp}
X_t=Y_t\;\;\; \textrm{and}\;\;\; \rho_{t+1}(X)\leq \rho_{t+1}(Y)\quad\Longrightarrow\quad \rt(X)\leq \rt(Y).
\end{equation}
\end{definition}
Note that a dynamic risk measure for processes $(\rt)_{t\in\T\cap\nat_0}$ is time consistent if and only if the corresponding dynamic convex risk measure for random variables $(\bar\rho_t)_{t\in\T\cap\nat_0}$ on $\LTp$ defined by (\ref{rrp}) is time consistent, that is, if $\bar\rho_{t+1}(X)\leq \bar\rho_{t+1}(Y)$ implies $\bar\rho_t(X)\leq \bar\rho_t(Y)$ for all $X,Y\in\LTp$ and all $t\in\T$, $t<T$. Criteria for time consistency of risk measures for random variables were studied intensively in the literature, see, e.g., \cite{dt5}, \cite{ks5}, \cite{adehk7}, \cite{fp6}, \cite{bn6}, \cite{bn8}, \cite{ap9} and the references therein. Using Proposition~\ref{corresp} we can translate these criteria into our present framework.

By \cite[Proposition 4.2]{fp6} applied to $\bar\rho$, time consistency \eqref{tcp} of $\rho$ is equivalent to {\it recursiveness}, that is
\begin{align}\label{recurs}
\rt(X)&=\rt(X_t1_{\{t\}}-\rho_{t+1}(X)1_{\T_{t+1}})\\
&=-X_t+\rt(-\rho_{t+1}(X-X_t)1_{\T_{t+1}})\nonumber.
\end{align}

If we restrict the conditional convex risk measure $\bar{\rho}_t$ to the space $L^\infty(\bar{\Omega},\bar{\F}_{t+1},\bar{P})$, the acceptance set is given by
\begin{eqnarray*}%\label{barats}
\bar{\A}_{t,t+1}&:=&\lk X\in L^\infty(\bar{\Omega},\bar{\F}_{t+1},\bar{P})\mk\bar{\rho}_t(X)\le0\;\,\bar{P}\text{-a.s.}\rk \nonumber\\
&=&\A_{t,t+1}+ L_{0,+}^{\infty}\times\ldots\times L_{t-1,+}^{\infty}\times\{0\}\times\ldots,
\end{eqnarray*}
where
\begin{equation*}
\mathcal{A}_{t,t+1}:=\{X\in\mathcal{R}_{t,t+1}^{\infty}\mk \rt(X)\leq 0\},\quad t\in\T,\quad t<T,
\end{equation*}
denotes the acceptance set of the risk measure for processes $\rt$ restricted to $\R^\infty_{t,t+1}$.
The corresponding one-step minimal penalty function for $\bar\rt$ takes the form
\begin{equation*}%\label{baralts}
\bar{\alpha}_{t,t+1}(\bar Q):=\qesp_{X\in\bar{\A}_{t,t+1}}E_{\bar Q}[-X\,|\,\bar\F_t\,]=\alpha_{t,t+1}(\bar{Q})1_{\{t,t+1\pke \}},\qquad \bar{Q}\in\map,
\end{equation*}
where the function $\alpha_{t,t+1}(\bar{Q})$ is given for $\bar{Q}=Q\otimes D=Q\otimes \gamma\in\map$ by
\begin{equation*}
\alpha_{t,t+1}(Q\otimes D)=\frac{1}{D_t}\qes_{X\in\A_{t,t+1}}E_Q\lp-\gamma_tX_t-D_{t+1}X_{t+1}\mk\F_t\rp,\quad t\in\T,\quad t<T,
\end{equation*}
due to Corollary~\ref{cor:condexp}. Note that the penalty functions $\alpha_t(Q\otimes D)$ and $\alpha_{t,t+1}(Q\otimes D)$ are only defined $Q$-a.s.\ on $\{D_t>0\}$. In the following we \emph{define} for $Q\otimes D\in\map$
\begin{equation*}%\label{extpen}
 \alpha_t(Q\otimes D):=\infty,\quad \alpha_{t,t+s}(Q\otimes D):=\infty\quad Q\text{-a.s. on}\;\{D_t=0\}
\end{equation*}
for all $t,s\ge 0$, and use henceforth the convention $0\cdot\infty:=0$.

The following result characterizes time consistency in terms of a splitting property of the acceptance sets and in terms of supermartingale properties of the penalty process and the dynamic risk measure. It translates \cite[Theorem 4.5]{fp6} and \cite[Theorem 17]{ap9} to our present framework. 
\begin{theorem}\label{eqcharp}
Let \rtf be a dynamic convex risk measure on $\R^\infty$ such that each $\rt$ is \ctse. Then the
following conditions are equivalent:
\begin{itemize}
\item[(i)] \rtf is time consistent;
\item[(ii)] $\A_t=\A_{t,t+1}+\A_{t+1}\;\;$\quad for all $t\in\T$,  $t<T$;
\item[(iii)] for all $t\in\T$,  $t<T$ and $\bar{Q}=Q\otimes D\in\map$
\begin{equation*}%\label{recda}
D_t\alpha_t(Q\otimes D)=D_t\alpha_{t,t+1}(Q\otimes D)+E_Q[D_{t+1}\alpha_{t+1}(Q\otimes D)\mk\F_t]\quad \qf;
\end{equation*}
\item[(iv)] for all $X\in\Ri$, $t\in\T$,  $t<T$, and $\bar{Q}=Q\otimes D\in\map$
\begin{equation*}%\label{smtcp}
E_Q[D_{t+1}(X_t+\rho_{t+1}(X)+\alpha_{t+1}(Q\otimes D))\mk \F_t]\le
D_t(X_t+\rho_t(X)+\alpha_t(Q\otimes D))\quad \qf.
\end{equation*}
\end{itemize}
Moreover, if there exists a probability measure $\bar P^*\approx \bar P$ on $(\bar\Omega, \bar\F)$ such that $\alpha_0(\bar P^*)<\infty$, condition (iv) stated only for the measures 
\begin{align}\label{qstar}
\bar\Q^*&:=\lk \bar Q\in\mep\mk \alpha_0(\bar Q)<\infty\rk\\
 & = \left\{Q\otimes \gamma\mk Q\in\mel,\; \gamma\in\Gamma^e(Q),\;\alpha_0(Q\otimes \gamma)<\infty\right\}\nonumber
\end{align}
already implies time consistency, and the robust representation \eqref{newreprinf} of $\rt$ also holds if the essential supremum is taken only over the set $\bar\Q^*$.
\end{theorem}
\begin{proof}
 Follows from \cite[Theorem 17]{ap9}  and \cite[Theorem 4.5]{fp6} applied to $\bar\rt$ defined in \eqref{rrp}  using Corollary~\ref{cor:condexp}.
\end{proof}

\begin{remark}
Equivalence of time consistency and (ii) of Theorem~\ref{eqcharp} holds without assuming continuity from above and was already proved in \cite[Theorem 4.6]{cdk6}. Characterizations of time consistency in terms of penalty functions as in condition (iii) are given in \cite[Theorem 4.19, Theorem 4.22]{cdk6}. However, the latter results use neither the decomposition of $\bar Q$ into a measure $Q$ and a discounting factor $D$, nor the one-step penalty functions $\alpha_{t,t+1}$. The role of $\alpha_{t,t+1}$ in condition (iii) is analogous to the corresponding characterization of time consistency of risk measures for random variables in \cite[Theorem 2.5]{bn6} and \cite[Theorem 4.5]{fp6}. In the same way, the supermartingale characterization (iv) of time consistency  translates the corresponding criterion from \cite[Theorem 4.5]{fp6} into our present framework.
\end{remark}
\begin{assumption}\label{cts}
From now on until the end of Section \ref{tc} we fix a \emph{time consistent} dynamic convex risk measure for processes \rtf such that each $\rt$ is \emph{continuous from above}.
\end{assumption}
In the following we use the notation
\[
\bar\Q_0:=\lk Q\otimes D\in\map\mk \alpha_0(Q\otimes D)<\infty\rk.
\]
\begin{corollary}\label{supermartingales}
\begin{enumerate}
\item For any $\bar Q=Q\otimes D\in\bar\Q_0$, the discounted penalty process $(D_t\alpha_t(Q\otimes D))_{t\in\T\cap\nat_0}$ is a nonnegative $Q$-supermartingale. Its Doob decomposition is given by the predictable process 
\[
 A_t^{Q,D}:=\sum_{k=0}^{t-1}D_k\alpha_{k,k+1}(Q\otimes D),\qquad t\in\T\cap\nat_0,
\]
i.e.,
\begin{equation}\label{mart:doob}
 M_t^{Q,D}:= D_t\alpha_t(Q\otimes D)+A_t^{Q,D},\qquad t\in\T\cap\nat_0,
\end{equation}
is a $Q$-martingale.
\item For all $X\in\R^{\infty}$ and all $\bar Q\in\bar\Q_0$, the process 
\begin{equation}\label{pw}
W_t^{Q,D}(X):=D_t\rt(X-X_t1_{\T_t})+\sum_{s=0}^tD_s(-\Delta X_s)+D_t\alpha_t(Q\otimes D),\quad t\in\T\cap\nat_0,
\end{equation}
is a $Q$-supermartingale.
\end{enumerate}
\end{corollary}

\subsection{Riesz decomposition of the penalty process and the appearance of bubbles}\label{rieszsec}
The following proposition characterizes the martingale $M^{Q,D}$ in the Doob decomposition of the $Q$-supermartingale $(D_t\alpha_t(Q\otimes D))\zt$ from Corollary~\ref{supermartingales}; it translates \cite[Proposition 21]{ap9} and \cite[Proposition 2.3.2]{ipen7} into our present context.

\begin{proposition}\label{rieszp}
The martingale $M^{Q,D}$ in \eqref{mart:doob} is of the form
\[
 M_t^{Q,D}=E_Q\left[\sum_{k=0}^{T-1}D_k\alpha_{k,k+1}(Q\otimes D)\mk\F_t\right]+N_t^{Q,D}\quad \qf,\quad t\in\T\cap\nat_0,
\]
where
\[
N_t^{Q,D}:=\left\{
\begin{array}{c@{\quad  \quad}l}
0 & \text{if $T<\infty$}\\
\displaystyle\lim_{s\to\infty}E_Q\left[D_s\alpha_s(Q\otimes D)\,|\,\F_t\,\right] & \text{if $T=\infty$}
\end{array}\right.\qquad \qf,\quad t\in\T\cap\nat_0,
\]
is a nonnegative $Q$-martingale. Thus the Riesz decomposition of the $Q$-supermartingale  $(D_t\alpha_t(Q\otimes D))$ into  a potential and a martingale takes the form
\begin{equation}\label{riesz}
D_t\alpha_t(Q\otimes D)=E_Q\left[\sum_{k=t}^{T-1}D_k\alpha_{k,k+1}(Q\otimes D)\,\big|\,\F_t\,\right]+N_t^{Q,D}\quad \qf,\quad t\in\T\cap\nat_0.
\end{equation}
\end{proposition}
\begin{proof} 
Property (iii) of Theorem~\ref{eqcharp} yields
\begin{equation}\label{lim}
D_t\alpha_t(\bar{Q})=E_Q\left[\,\sum_{k=t}^{t+s-1}D_k\alpha_{k,k+1}(\bar{Q})\,\big|\,\F_t\,\right]+E_Q[\,D_{t+s}\alpha_{t+s}(\bar{Q})\,|\,\F_t\,]\quad \qf
\end{equation}
for all $t,s\in\nat_0$ s.t. $t+s\in\T$ and all $\bar{Q}\in\map$. For $T<\infty$ the claim is obvious, since $\alpha_{T}(\bar Q)=0$ $P$-a.s.. For $T=\infty$, by monotonicity there exists the limit
\[
S_t^{Q,D}= \lim_{s\to\infty}E_Q\left[\,\sum_{k=t}^{s}D_k\alpha_{k,k+1}(\bar{Q})\,\big|\,\F_t\,\right]=E_Q\left[\,\sum_{k=t}^{\infty}D_k\alpha_{k,k+1}(\bar{Q})\,\big|\,\F_t\,\right]\quad \qf
\]
for all $t\in\T\cap\nat_0$, where we have used the monotone convergence theorem for the second equality. Thus  \eqref{lim} implies existence of 
\[
N_t^{Q,D}= \lim_{s\to\infty}E_Q[\,D_{t+s}\alpha_{t+s}(\bar{Q})\,|\,\F_t\,]\quad \qf,\quad t\in\T\cap\nat_0
\]
and
\[
D_t\alpha_t(\bar{Q})=S_t^{Q,D}+N_t^{Q,D}\quad \qf,\quad t\in\T\cap\nat_0.
\]
The process $(S_t^{Q,D})$ is a $Q$-potential. Indeed, 
\begin{equation*}%\label{fin}
E_Q[\,S_{t}^{Q,D}\,]\le E_Q\left[\,\sum_{k=0}^{\infty}D_k\alpha_{k,k+1}(\bar{Q})\,\big|\,\F_t\,\right]\le\alpha_0(\bar{Q})<\infty
\end{equation*}
and $E_Q[\,S_{t+1}^{Q,D}\,|\,\F_t\,]\le S_t^{Q,D}\;\,Q$-a.s. for all $t\in\T\cap\nat_0$ by definition. Moreover, monotone convergence implies
\[
\lim_{t\to\infty}E_Q[\,S_t^{Q,D}\,]=E_Q\left[\,\lim_{t\to\infty}\sum_{k=t}^{\infty}D_k\alpha_{k,k+1}(\bar{Q})\,\right]=0\quad\qf.
\]
The process $(N_t^{Q,D})$ is a nonnegative $Q$-martingale, since 
\begin{align*}
E_Q[N_{t+1}^{Q,D}-N_{t}^{Q,D}|\F_t]&=E_Q[D_{t+1}\alpha_{t+1}(\bar{Q})|\F_t]-D_t\alpha_t(\bar{Q})-E_Q[S_{t+1}^{Q,D}-S_t^{Q,D}|\F_t]\\
&=D_t\alpha_{t,t+1}(\bar{Q})-D_t\alpha_{t,t+1}(\bar{Q})=0\qquad Q\mbox{-a.s.}
\end{align*}
for all $t\in\T\cap\nat_0$ by property (iii) of Theorem~\ref{eqcharp} and the definition of $(S_t^{Q,D})$.
\end{proof}
The nonnegative martingale $N^{Q,D}$, which may appear in the decomposition \eqref{riesz} of the penalty process for $T=\infty$, plays the role of a ``bubble". Indeed, it appears on top of the ``fundamental" component which is given by the potential $S^{Q,D}$ generated by the one-step penalties, and this additional penalization causes an excessive neglect of the model $Q\otimes D$ in assessing the risk. As a result, asymptotic safety breaks down under the model $Q\otimes D$, as explained in the next section.

\subsection{Asymptotic safety and asymptotic precision}\label{asympsec}
In this section we discuss the asymptotic properties of dynamic convex risk measures for processes. Throughout this section  we consider the case $\T=\nat_0\cup\{\infty\}$. We fix a time consistent dynamic convex risk measure for processes $(\rt)_{t\in\nat_0}$. As before, $(\bar\rt)_{t\in\nat_0}$ denotes the corresponding time consistent dynamic convex risk measure for random variables on product space given by \eqref{rrp}.

Let $\bar Q=Q\otimes\gamma=Q\otimes D\in\bar Q_0$, and let us focus on the behavior of  $(\bar\rt)_{t\in\nat_0}$ under $\bar Q$. The measure $\bar Q$ will now play the same role as the reference measure $P$ in \cite[Section 5]{fp6}. In particular, the assumption $\Q^*\ne\emptyset$ from \cite[Section 5]{fp6} is satisfied for $\bar Q$, since $\bar Q\in\bar Q_0$. 

The results in \cite{fp6} imply the existence of the limits
\[
\bar\alpha_\infty(\bar Q):=\lim_{t\to\infty}\bar\alpha_t(\bar Q)\quad\text{and}\quad \bar\rho_\infty(X):=\lim_{t\to\infty}\bar\rt(X)\quad\bar Q \text{-a.s.}
\]
for all $X\in\Ri$.  Due to \eqref{rrp} and \eqref{corpen}, we have
\begin{equation}\label{as1}
\bar\rho_\infty(X)=-XI_{\nat_0}+\rho_\infty(X)I_{\{\infty\}}\quad \text{and}\quad \bar\alpha_\infty(\bar Q)=\alpha_\infty(\bar Q)I_{\{\infty\}}\quad\bar Q \text{-a.s.},
\end{equation}
where
\[
\rho_\infty(X):=\lim_{t\to\infty}\rt(X)\quad \text{and}\quad \alpha_\infty(\bar Q)=\lim_{t\to\infty}\alpha_t(\bar Q)\quad Q \text{-a.s.}\;\,\text{on}\;\,\{D_\infty>0\}
\]
by 3 of Remark~\ref{rmku}.

\begin{definition}\label{def:assymptoticsafety}
We call a dynamic convex risk measure for processes $(\rt)_{t\in\nat_0}$ \emph{asymptotically safe under the model $\bar Q=Q\otimes D$} if the limiting capital requirement $\rho_\infty(X)$ covers the final loss $-X_\infty$, i.e. 
 \[
 \rho_\infty(X)\ge-X_\infty\quad Q \text{-a.s.}\;\,\text{on}\;\,\{D_\infty>0\}
 \]
for any $X\in\Ri$.
\end{definition}
Note that due to \eqref{as1} asymptotic safety of $(\rt)_{t\in\nat_0}$ is equivalent to the condition
\[
\bar\rho_\infty(X)\ge-X\quad \bar Q\text{-a.s.},
\]
i.e., to asymptotic safety of $(\bar\rt)_{t\in\nat_0}$ in the sense of \cite[Definition 5.2]{fp6}. 

The following result translates \cite[Theorem 5.4]{fp6} and \cite[Corollary 3.1.5]{ipen7} to our present setting. It characterizes asymptotic safety by the absence of bubbles in the penalty process. This is plausible since, as we saw in Subsection~\ref{rieszsec}, such bubbles reflect an excessive neglect of models which may be relevant for the risk assessment.
\begin{theorem}\label{th:assymptoticsafety}
For a dynamic convex risk measure for processes $(\rt)_{t\in\nat_0}$ and for any model $\bar Q=Q\otimes D\in\bar\Q_0$ the following conditions are equivalent:
\begin{enumerate}
 \item $(\rt)$ is asymptotically safe under the model $\bar Q$;
\item the model $\bar Q$ has no bubble, i.e., the martingale $N^{Q,D}$ in the Riesz decomposition \eqref{riesz} of the discounted penalty process $(D_t\alpha_t(\bar Q))_{t\in\nat_0}$ vanishes;
\item the discounted penalty process $(D_t\alpha_t(\bar Q))_{t\in\nat_0}$ is a $Q$-potential;
\item no model $\bar R\ll\bar Q$ with $\alpha_0(\bar R)<\infty$ admits bubbles.
\end{enumerate}
\end{theorem}
\begin{proof}
Properties 2 and 3 are equivalent by \eqref{riesz}, and obviously 4 implies 2. \\
To prove $1 \Leftrightarrow 2$ we use \cite[Theorem 5.4]{fp6}. There it was shown that $(\bar\rt)$ is asymptotically safe under $\bar Q$ if and only if $\bar\alpha_\infty(\bar Q)=0$ $\bar Q$-a.s.\ and in $L^1(\bar Q)$. By Corollary~\ref{cor:condexp}, \eqref{corpen}, and \eqref{gammatodisc2} we have 
\[
 E_{\bar Q}[\bar\alpha_t(\bar Q)]=E_Q\left[ \sum_{s\in\T_t}\gamma_s\alpha_t(\bar Q)\right]=E_Q\lp D_t\alpha_t(\bar Q)\rp.
\]
Thus $\bar\alpha_t(\bar Q)\to0$ in $L^1(\bar Q)$ if and only if $D_t\alpha_t(\bar Q)\to0$ in $L^1(Q)$. This is equivalent to $N^{Q,D}\equiv0$, since the bubble $N^{Q,D}=(N_t^{Q,D})_{t\in\nat_0}$ is a nonnegative $Q$-martingale with $N_0^{Q,D}=\lim_{t\to\infty}E_Q\lp D_t\alpha_t(\bar Q)\rp$. Due to \eqref{riesz}, $N^{Q,D}\equiv0$ also implies $\alpha_\infty(\bar Q)=0$ $Q$-a.s.\ on $\{D_\infty>0\}$, thus $\bar\alpha_\infty(\bar Q)=0$ $\bar Q$-a.s. by \eqref{as1}. \\
To prove $2 \Rightarrow 4$ note that asymptotic safety under $\bar Q$ implies asymptotic safety under any model $\bar R\ll\bar Q$ with $\alpha_0(\bar R)<\infty$, thus no model $\bar R$ admits bubbles by the same reasoning as above. 
\end{proof}

\begin{definition}\label{def:assymptoticpre}
 We call a dynamic convex risk measure for processes $(\rt)_{t\in\nat_0}$ \emph{asymptotically precise under the model $\bar Q=Q\otimes D\in\bar\Q_0$} if 
 \[
 \rho_\infty(X)=-X_\infty\quad Q \text{-a.s.}\;\,\text{on}\;\,\{D_\infty>0\}
 \]
for any $X\in\Ri$.
\end{definition}
By \eqref{as1}, asymptotic precision of $(\rt)$ is equivalent to asymptotic precision of $(\bar \rt)$ in the sense of \cite[Definition 5.9]{fp6}.  The following corresponds to  \cite[Lemma 2.7]{nau7}.
\begin{lemma}\label{felix}
 A dynamic convex risk measure $(\rt)_{t\in\nat_0}$ is asymptotically precise under the model $\bar Q=Q\otimes D\in\bar\Q_0$ if and only if
\[
 \rho_\infty(X)\le-X_\infty\quad Q\text{-a.s}\;\,\text{on}\;\,\{D_\infty>0\}\quad\text{for all}\quad X\in\Ri.
\]
\end{lemma}
\begin{proof}
 By \cite[Lemma 5.1]{fp6} the functional $\bar\rho_\infty$ is convex and normalized. This implies
\[
 \bar\rho_\infty(X)\ge-\bar\rho_\infty(-X)\quad\text{for all}\quad X\in\Ri.
\]
Thus we obtain
\[
-X\ge \bar\rho_\infty(X)\ge-\bar\rho_\infty(-X)\ge-X\quad\bar Q\text{-a.s}\quad\text{for all}\quad X\in\Ri,
\]
which is equivalent to $\rho_\infty(X)=X_\infty$ $Q$-a.s.\ on $\{D_\infty>0\}$ by \eqref{as1}. 
\end{proof}

The following result translates \cite[Proposition 5.11]{fp6} to our present setting.
\begin{proposition}\label{asymptoticpre}
Assume that for each $X\in\Ri$ the supremum in the robust representation \eqref{newreprinf} of $\rho_0(X)$ is attained by some ``worst case'' measure $Q^X\otimes\gamma^X=\bar Q^X$, such that  $\bar Q^X\approx\bar Q$. Then $(\rt)_{t\in\nat_0}$ is asymptotically precise under $\bar Q$.
\end{proposition}
\begin{proof}
Since $\rho_0(X)=\bar\rho_0(X)$, $\bar{Q}^X$ is also a worst case measure for $\bar\rho_0(X)$.  By \cite[Proposition 18]{ap9}, the measure $\bar{Q}^X$ is then a worst case measure for $X$ at all times $t\in\nat_0$, i.e.,
\[
 \bar\rt(X)=E_{\bar Q^X}\lp -X|\bar\F_t\rp-\bar\alpha_t(\bar Q^X)\quad \bar Q\text{-a.s.}\quad\forall\;t\in\nat_0,
\]
and in particular $\bar{Q}^X\in\bar\Q_0$. By martingale convergence,
\[
 \bar\rho_\infty(X)=-X-\bar\alpha_\infty(\bar Q^X)\quad \bar Q\text{-a.s.},
\]
which is equivalent to 
\[
 \rho_\infty(X)=-X_\infty-\alpha_\infty(\bar Q^X)\quad Q\text{-a.s.}\;\,\text{on}\;\,\{D_\infty>0\}
\]
due to \eqref{as1}. Asymptotic precision of $(\rt)$ now follows from Lemma~\ref{felix}, since $\alpha_\infty(\bar Q^X)\ge0$ $Q$-a.s.\ on $\{D_\infty>0\}$.
\end{proof}

\subsection{A maximal inequality for the capital requirements}\label{maxsec}
For $X\in\R^{\infty}$ and $Q\otimes D\in\map$, we can interpret
\[
F_t^{Q,D}(X):=E_Q\lp-\sum_{s\in\T_t}\frac{\gamma_s}{D_t}X_s\mk\F_t\rp-\alpha_t(Q\otimes \gamma)\quad\text{on}\;\{D_t>0\}
\]
as a risk evaluation of the cash flow $X$  at time $t\in\T\cap\nat_0$, using the specific model $Q$ and the specific discounting process $D$. The next proposition provides, from the point of view of the model $Q$, a maximal inequality for the excess of the required capital $\rt(X)$ over the risk evaluation $F_t^{Q,D}(X)$.

\begin{proposition}
For $Q\otimes D\in\map$, $X\in\R^{\infty}$, and $c>0$ we have
\begin{equation}\label{maxin}
Q\left(\sup_{t\in\T\cap\nat_0}\left\{D_t\left(\rt(X)-F_t^{Q,D}(X)\right)\right\}\ge c\right)\le\frac{\rho_0(X)-F_0^{Q,D}(X)}{c}.
\end{equation}
\end{proposition}
\begin{proof} Fix $Q\otimes D\in\map$. If $\alpha_0(Q\otimes D)=\infty$, then the inequality \eqref{maxin} holds trivially. Assume that $\alpha_0(Q\otimes D)<\infty$.
By 2) of Corollary~\ref{cororep} we have
\[
\rt(X)\geq E_Q\lp-\sum_{s\in\T_t} \frac{\gamma_s}{D_t}X_s\mk\F_t\rp-\alpha_t(Q\otimes \gamma)=F_t^{Q,D}(X)\quad \qf\; \textrm{on $\{D_t>0\}$}.
\]
Thus the $Q$-supermartingale $W^{Q,D}(X)$ defined in \eqref{pw} satisfies
\begin{equation*}
W_t^{Q,D}(X)\ge - E_Q\lp\sum_{s\in\T} \gamma_sX_s\mk\F_t\rp\qquad\qf\quad \textrm{on $\{D_t>0\}$}.
\end{equation*}
On $\{D_t=0\}=\{ D_s=0\; \forall s\in\T_t\}$, we have $W_t^{Q,D}(X)=-\sum_{s=0}^{t-1} D_s\Delta X_s$.
Therefore, the process 
\[
Y_t^{Q,D}(X):=D_t\left(\rt(X)-F_t^{Q,D}(X)\right)=W_t^{Q,D}(X)+E_Q\lp\sum_{s\in\T} \gamma_sX_s\mk\F_t\rp,\quad t\in\T\cap\nat_0,
\]
is a nonnegative $Q$-supermartingale, and \eqref{maxin} follows by a classical maximal inequality; cf., e.g., \cite[Theorem VII.3.1]{shi96}.
\end{proof}

\subsection{The coherent case}\label{cohsec}
Due to positive homogeneity of a coherent risk measure, the penalty function can only take values $0$ or $\infty$, and thus a coherent risk measure for processes $\rt$ is continuous from above if and only if it admits the robust representation
\begin{equation}\label{repcoh}
\rt(X)=\es_{Q\otimes\gamma\in\Q_t^0}E_Q\lp-\sum_{s\in\T_t}\gamma_sX_s\mk\F_t\rp,  \qquad X\in\R_t^\infty,
\end{equation}
where 
\[
 \Q_t^0:=\lk \bar Q\in\bar\Q_t\mk \alpha_t(\bar Q)=0\rk.
\]

The next theorem reformulates properties (iii) and (iv) of Theorem~\ref{eqcharp} in the coherent case. This involves a translation of the notions of \emph{pasting} of measures and \emph{stability} of sets as used in \cite{adehk7}, \cite{d6}, \cite{fp6} in context of coherent risk measures for random variables to our present framework. 

For $\bar{Q}^1,\bar{Q}^2\in\map$ such that $\bar Q^1\ll\bar Q^2$ on ${\bar\F_t}$ and for $B\in\bar\F_t$ we denote by $\bar{Q}^1\oplus^{t}_B\bar{Q}^2$ the \emph{pasting} of $\bar{Q}^1$ and $\bar{Q}^2$ in $t$ via $B$, i.e., the probability measure on $(\bar\Omega, \bar\F)$ defined by
\[
\bar{Q}^1\oplus^{t}_B\bar{Q}^2(A)=E_{\bar{Q}^1}\left[E_{\bar{Q}^2}[1_A|\bar{\F}_{t}]1_B+1_{B^c}1_A\right],\qquad  A\in\bar\F.
\]
Theorem~\ref{propq} yields the decomposition $\bar Q^i=Q^i\otimes D^i$,  $i=1,2$ with $Q^1\ll Q^2$ on $\F_t$.  Then 
\[
\bar{Q}^1\oplus^{t}_B\bar{Q}^2=Q^0\otimes D^0,
\]
 where $B_t=\{\omega|(\omega, t)\in B\}\in\F_t$, and $Q^0=Q^1\oplus^{t}_{B_t}Q^2$, i.e.
\[
Q^0(A)=E_{Q^1}\left[E_{Q^2}[1_A|\F_{t}]1_{B_t}+1_{B_t^c}1_A\right],\qquad A\in\F_T,
\]
and
\begin{eqnarray*}
\gamma^{0}_u = \left\{ \begin{array}{c@{\quad\quad}l} \gamma^{1}_u & u=0\pk t-1\\
D_{t}^1\dfrac{\gamma^{2}_u}{D_{t}^2}1_{\{D^2_t>0\}}1_{B_t}+\gamma^{1}_u1_{B_t^c}  & u\in\T_t.
\end{array} \right.
\end{eqnarray*}
Here $\gamma^{i}$ and $D^i$ are related to each other via \eqref{gammatodisc} and \eqref{disctogamma} for $i=0,1,2$. Note that $Q^0\in\mal, D^0\in\D(Q^0)$, in other words, the pasting of $Q^1\otimes D^1$ with $Q^2\otimes D^2$ admits a decomposition with the pasting of $Q^1$ with $Q^2$ and the pasting of $D^1$ with $D^2$. 
\begin{definition}
We call a set $\bar\Q\subseteq\map$ \emph{stable} if, whenever $\bar Q^1, \bar Q^2\in\bar\Q$ and $\bar Q^1\ll\bar Q^2$ on ${\bar\F_t}$, the pasting of $\bar{Q}^1$ and $\bar{Q}^2$ in $t$ via $B$ belongs to $\bar\Q$ for every $B\in\bar\F_t$ and all $t\in\T\cap\nat_0$. 
\end{definition}

We associate to any $\bar{Q}\in\map$ the sets
\[
\Q_t^0(\bar{Q}) = \lk \bar{R}\in\map\mk \bar{R}=\bar{Q}|_{\bar{\F}_t},\; \bar\alpha_t(\bar{R})=0\, \bar{Q}\textrm{-a.s.} \rk,
\]
and
\[
\Q_{t,t+s}^0(\bar{Q})=\lk \bar{R}\ll\bar{P}|_{\bar{\F}_{t+s}}\mk \bar{R}=\bar{Q}|_{\F_t},\; \bar\alpha_{t,t+s}(\bar{R})=0\; \bar{Q}\textrm{-a.s.} \rk.
\]

The notion of pasting corresponds to \emph{concatenation} defined in \cite[Definition 4.10]{cdk6} on $\A^1$, and the following corollary is related to \cite[Theorem 4.13, Corollary 4.14]{cdk6}.

\begin{theorem}\label{cormp}
Suppose that the dynamic risk measure \rtf is coherent. Then the following conditions are equivalent:
\begin{itemize}
\item[1.] \rtf is time consistent
\item[2.] For all $t\in\T$, $t<T$ and $\bar{Q}\in\map$,
\begin{equation*}
\Q_t^0(\bar{Q}) = \lk \bar{Q}^1\oplus^{t+1}_\Omega \bar{Q}^2\mk \bar{Q}^1\in\Q_{t,t+1}^0(\bar{Q}), \;\bar{Q}^2\in\Q_{t+1}^0(\bar{Q}^1)\rk.
\end{equation*}
\item[3.] For all $t\in\T$, $t<T$, $X\in\R^\infty$ and $\bar{Q}=Q\otimes D\in\map$ such that $\alpha_t(\bar{Q})=0\;\qf$ on $\{D_t>0\}$,
\[
E_Q[D_{t+1}(X_t+\rho_{t+1}(X))\,|\,\F_t] \le D_t(X_t+\rho_t(X))\quad\text{and}\quad \alpha_{t+1}(\bar{Q})=0\;\qf \;\;\text{on $\{D_{t+1}>0\}$}.
\] 
\end{itemize}
Moreover, if the set $\bar\Q^*$ defined in \eqref{qstar} is not empty, then time consistency is equivalent to each of the following conditions:
\begin{itemize}
\item[4.] The set $\bar\Q^*$ is stable, and $\rt$ has the representation
\begin{equation}\label{robdarcoh}
 \rt(X)=\es_{Q\otimes D\gamma\in\bar\Q^*}\frac{1}{D_t}E_Q\lp -\sum_{s\in\T_t}\gamma_s X_s\mk \F_t\rp
\end{equation}
for all $X\in\R^\infty$ and $t\in\T\cap\nat_0$.

\item[5.] The representation \eqref{robdarcoh} holds for all $t\in\T\cap\nat_0$ and all $X\in\R^\infty$, and the process
\begin{equation*}
D_t\rt(X-X_t1_{\T_t})-\sum_{s=0}^tD_s\Delta X_s,\qquad t\in\T\cap\nat_0
\end{equation*}
is a $Q$-supermartingale for all $\bar Q=Q\otimes D\in\bar\Q^*$.
\end{itemize}
\end{theorem}
\begin{proof} Follows by applying \cite[Corollary 23]{ap9} and \cite[Corollary 4.12]{fp6} to $\bar{\rho}$ defined in \eqref{rrp} and using Corollary~\ref{cor:condexp}.
 \end{proof}

\begin{remark}
Note that due to Theorem~\ref{th:assymptoticsafety} coherence implies that the risk measure is asymptotically safe under any model $\bar Q=Q\otimes D\in\Q_0^0$. Indeed, by 1 of Corollary~\ref{supermartingales}, $(D_t\alpha_t(\bar Q))_{t\in\nat_0}$ is a nonnegative $Q$-supermartingale beginning at $0$, and hence it vanishes. In particular there are no bubbles in the coherent case. 
\end{remark}

\subsection{Weaker notions of time consistency}\label{weaksec}
In this section we characterize some weaker notions of time consistency that appeared in \cite{Weber}, \cite{adehk7}, \cite{burg}, \cite{tu6}, \cite{Samuel}, \cite{ros7} \cite{ipen7}, \cite{ap9} in context of risk measures for random variables. 

\begin{definition}
A dynamic convex risk measure $(\rt)_{t\in\T\cap\nat_0}$ on $\R^{\infty}$ is called \emph{acceptance consistent} (resp.\ \emph{rejection consistent}) if
for all $t\in\T$ such that $t<T$, and for all $X\in\R^{\infty}$
\begin{equation}\label{recmp}
\rt(X)\leq \rt(X_t1_{\{t\}}-\rho_{t+1}(X)1_{\T_{t+1}})\;\;\; P\text{-a.s.}\qquad (\mbox{resp.} \geq).
\end{equation}
\end{definition}

Note that $(\rt)_{t\in\T\cap\nat_0}$ is acceptance (resp.\ rejection) consistent if and only if the corresponding dynamic convex risk measure $(\bar\rho_t)_{t\in\T\cap\nat_0}$ on $\LTp$ defined in (\ref{rrp}) is acceptance (resp.\ rejection) consistent in the sense of \cite[Proposition 24]{ap9}; cf.\ also \cite[Definition 3.1]{Samuel}. Similar to Theorem~\ref{eqcharp}, the following theorem translates characterizations of acceptance and rejection consistency from \cite[Theorem 27, Proposition 29]{ap9} to our present framework.

\begin{theorem}\label{eqcharweak}
Let \rtf be a dynamic convex risk measure on $\R^\infty$. Then the
following conditions are equivalent:
\begin{itemize}
\item[(i)] \rtf is acceptance (resp.\ rejection) consistent;
\item[(ii)] $\A_t\supseteq A_{t,t+1}+\A_{t+1}\;\;$\quad (resp.\ $\subseteq$)\quad for all $t\in\T$, $t<T$;
\item[(iii)] for all $X\in\Ri$, $t\in\T$, $t<T$, and all $\bar{Q}=Q\otimes D\in\map$
\begin{equation*}%\label{recda}
D_t\alpha_t(Q\otimes D)\ge D_t\alpha_{t,t+1}(Q\otimes D)+E_Q[D_{t+1}\alpha_{t+1}(Q\otimes D)\mk\F_t]\quad\textrm{(resp.\ $\le$)}\quad \qf.
\end{equation*}
\end{itemize}

Moreover, rejection consistency is equivalent to the following: 
\begin{itemize}
\item[(iv)] for all $t\in\T$, $t<T$, and all $\bar{Q}=Q\otimes D\in\map$
\[
E_Q[D_{t+1}(X_t+\rho_{t+1}(X))\,|\,\F_t]\le
D_t(X_t+\rho_t(X))+\alpha_{t,t+1}(Q\otimes D)\quad \qf.
\] 
\end{itemize}
\end{theorem}
\begin{corollary}
\begin{enumerate}
\item For a rejection consistent dynamic convex risk measure, property (iv) of Theorem~\ref{eqcharweak} implies that the process
 \[
  D_t\rt(X-X_t1_{\T_t})-\sum_{s=0}^tD_s\Delta X_s-\sum_{s=0}^{t-1}D_s\alpha_{s,s+1}(\bar Q),\qquad t\in\T\cap\nat_0,
 \]
is a $Q$-supermartingale for all $\bar Q=Q\otimes D\in\map$ such that $E_Q[\sum_{s=0}^t D_s\alpha_{s,s+1}(\bar Q)]<\infty$ for all $t\in\T$, $t<T$.

\item For an acceptance consistent dynamic convex risk measure, property (iii) of Theorem~\ref{eqcharweak} implies that the discounted penalty process $(D_t\alpha_{t}(\bar Q))_{t\in\T\cap\nat_0}$ is a $Q$-supermartingale for all $\bar Q=Q\otimes D\in\bar\Q_0$. 
\end{enumerate}
\end{corollary}

The following definition translates the weak notion of time consistency from \cite{Weber}, \cite{adehk7}, \cite{burg}, \cite{tu6}, \cite{ros7}, \cite{ap9} to our present framework. 
\begin{definition}
A dynamic convex risk measure $(\rt)_{t\in\T\cap\nat_0}$ on $\R^{\infty}$ is called \emph{weakly acceptance consistent}  if
\begin{equation*}
X_t=0\;\;\; \textrm{and}\;\;\; \rho_{t+1}(X)\le0\quad\Longrightarrow\quad \rt(X)\leq 0
\end{equation*}
for all $t\in\T$ such that $t<T$ and for all $X\in\R^{\infty}$.
\end{definition}

\begin{proposition}\label{weaktcp}
For a dynamic convex risk measure \rtf  the following properties are equivalent: 
\begin{enumerate}
\item \rtf is weakly acceptance consistent;
\item $\A_{t+1}\subseteq\A_t\quad$ for all $t\in\T$, $t<T$;
\item for all $t\in\T$, $t<T$, and all $\bar{Q}=Q\otimes D\in\map$
\begin{equation*}%\label{supmgp}
E_Q[D_{t+1}\alpha_{t+1}(Q\otimes D)\mk\F_t]\le D_t\alpha_t(Q\otimes D)\quad \qf.
\end{equation*}
\end{enumerate} 
In particular, if $(\rt)$  is weakly acceptance consistent, then the discounted penalty process $(D_t\alpha_t(\bar Q))_{t\in\T\cap\nat_0}$ is a nonnegative $Q$-supermartingale for each $\bar{Q}=Q\otimes D\in\bar\Q_0$.
\end{proposition}
\begin{proof}
 Follows from \cite[Proposition 33]{ap9} applied to $\bar\rt$ defined in \eqref{rrp}.
\end{proof}

\section{Cash subadditivity and calibration to num\'eraires}\label{cashadsec}
As noted after Definition~\ref{defrmp}, cash invariance of risk measures for processes differs from the corresponding property of risk measures for random variables, since it takes into account the timing of the payment. This aspect can be made precise using the notion of cash subadditivity. Cash subadditivity was introduced by El Karoui and Ravanelli~\cite{er08} in the context of risk measures for random variables in order to account for discounting ambiguity. It will be shown in Proposition~\ref{propcash}, and it is also apparent from the robust representation given in Subsection~\ref{robsec}, that every risk measure for processes is cash subadditive. Thus risk measures for processes provide a natural framework to capture uncertainty about the time value of money, and a systematic approach to the issue of discounting ambiguity.

\subsection{Cash subadditivity}
\begin{definition}
A conditional convex risk measure for processes $\rt$ is called
\begin{itemize}
\item \emph{cash subadditive} if
\begin{equation}\label{subcash}
\rt(X+m1_{\T_{t+s}})\geq\rt(X)-m,\;\;\; \forall\; s>0,\;\; \forall\;  m\in\Lts,\;\; m\geq 0;
\end{equation}
\item \emph{cash additive} at time $t+s$, with $s>0$ and $t+s\in\T$, if
\begin{equation*}%\label{cash}
\rt(X+m1_{\T_{t+s}})=\rt(X)-m,\quad \forall\; m\in\Lts,
\end{equation*}
\item \emph{cash additive} if it is cash additive at all times $s\in\T_{t+1}$.
\end{itemize}
\end{definition}
\begin{remark}
 Note that \eqref{subcash} is equivalent to 
\begin{equation*}%\label{subcash}
\rt(X+m1_{\T_{t+s}})\leq\rt(X)-m,\;\;\; \forall\; s>0,\;\;\forall\;  m\in\Lts,\;\; m\leq 0,
\end{equation*}
since $\rt(X)=\rt(X+m1_{\T_{t+s}}-m1_{\T_{t+s}})$.
\end{remark}

Cash subadditive risk measures account for the timing of the payment in the sense that the risk is reduced by having positive inflows earlier and negative ones later. Other equivalent characterizations of cash subadditivity can be found in \cite[Section 3.1]{er08}.

As noted in \cite{ck6} in the time consistent case, cash subadditivity is an immediate consequence of the basic properties of a conditional risk measure for processes.  

\begin{proposition}\label{propcash}
Every conditional convex risk measure for processes $\rt$  is cash subadditive.
\end{proposition}

\begin{proof}
Cash subadditivity follows straightforward from monotonicity and cash invariance of $\rt$:
\[
\rt(X)-m=\rt(X+m1_{\T_t})\leq \rt(X+m1_{\T_{t+s}}),\quad \forall s>0,\;\,\forall\;  m\in\Lts,\;\, m\ge0.
\]
\end{proof}

Cash subadditivity of risk measures for processes is also apparent from the robust representation given in Subsection~\ref{robsec} due to the appearance of the discounting factors.

\begin{remark}\label{redekr}
In particular, for $T<\infty$ or $\T=\nat_0\cup\{\infty\}$, every risk measure for processes restricted to the space $\{X\in\Ri | X_t=0,\; t<T\}$ defines a cash subadditive risk measure on $L^{\infty}$ in the sense of \cite[Definition 3.1]{er08}.
\end{remark}

\begin{remark}\label{rmkca}
For $\T=\nat_0$, a conditional convex risk measure for processes $\rt$ that is continuous from above cannot be cash additive. Indeed, if $\rt$ is cash additive at $t+s$ for all $s>0$, continuity from above implies for $X\in\R^{\infty}$ and $m\in\Lt, m>0,$
\[
-m+\rt(X)=\rt(X+m1_{\T_{t+s}})\nearrow \rt(X)\quad \text{with}\quad s\to\infty,
\]
which is absurd.  The interpretation of this result is clear: If we are indifferent between having an amount of money today or tomorrow or at any future time, then any payment can be shifted  from one date to the next, and so it would never appear.
\end{remark}

The following proposition describes the interplay between time consistency and cash additivity.

\begin{proposition}\label{infca}
Let $(\rt)\zt$ be a time consistent dynamic convex risk measure on $\mathcal{R}^{\infty}$ such that each $\rt$ is cash additive at time $t+1$. Then each $\rt$ is cash additive.
\end{proposition}
\begin{proof}
Follows by induction using one-step cash additivity and recursiveness \eqref{recurs}.
\end{proof}

In view of Proposition~\ref{infca} and Remark~\ref{rmkca} we obtain the following result.
\begin{corollary}\label{coratc}
For $\T=\nat_0$, a dynamic convex risk measure $(\rt)_{t\in\nat_0}$ on $\R^{\infty}$ such that each $\rt$ is continuous from above and cash additive at time $t+1$ cannot be time consistent.
\end{corollary}

\begin{remark}
Corollary~\ref{coratc} and Remark~\ref{rmkca} heavily depend on the assumption of continuity from above, which was formulated as a global property. For $\T=\nat_0$, the corollary in fact suggests to replace global continuity from above by a local version; this is done in \cite{fp10}.
\end{remark}

\subsection{Calibration to num\'eraires}
Cash additivity can be seen as additivity with respect to the num\'eraire $1$. In this section we discuss additivity with respect to  other possible num\'eraires. To this end we formulate conditional versions of some results from \cite{er08}.

\begin{assumption}
In the rest of Section~\ref{cashadsec} all conditional convex risk measures $\rt$ are assumed to be continuous from above.
\end{assumption}
As usual, we denote by $\alpha_t$ the minimal penalty function of $\rt$, and for $t\in\T\cap\nat_0$ we define
\[
 \Q_t^{\alpha}:=\lk Q\in\Q_t\mk \alpha_t(Q)<\infty \rk,\quad\bar\Q_t^{\alpha}:=\lk \bar Q\in\bar\Q_t\mk \alpha_t(\bar Q)<\infty\rk, 
\]
where
\begin{equation*}
 \Q_t:=\lk Q\in\ma\mk Q=P\,\;\text{on}\;\,\F_t\rk,
\end{equation*}
and $\bar\Q_t$ is defined in \eqref{bqt}.

The following lemma is a conditional version of \cite[Lemma 2.3]{er08}.
\begin{lemma}\label{lemlin}
Let $\rt\,:\,\LTs\,\rightarrow\,\Lts$ be a conditional convex risk measure for random variables, and let $N\in\LT$. Then the following conditions are equivalent:
\begin{itemize}
\item[(i)]  $\rt(\lambda_tN)=\lambda_t\rt(N)$ for all $\lambda_t\in\Lts$;
\item[(ii)] $E_Q[-N\,|\,\F_t\,]=\rt(N)$ for all $Q\in\Q_t^\alpha$;
\item[(iii)] $\rt(X+\lambda_tN)=\rt(X)+\lambda_t\rt(N)$ for all $X\in\LTs$ and all $\lambda_t\in\Lts$.

\end{itemize}
\end{lemma}

\begin{proof}
$(i)\Rightarrow(ii)$.\, (i) and the robust representation \eqref{rd:fp} imply for each $\lambda_t\in\Lts$ and $Q\in\Q_t$ 
\[
\lambda_t\rt(N)=\rt(\lambda_tN)\geq \lambda_t E_Q[-N|\F_t]-\alpha_t(Q). 
\]
If $\alpha_t(Q)<\infty$, we have  $\alpha_t(Q)\geq -\lambda_t(E_Q[N|\F_t]+\rt(N))$ for any $\lambda_t\in\Lts$, thus $\rt(N)=E_Q[-N|\F_t]$.\\
$(ii)\Rightarrow(iii)$ follows from the robust representation \eqref{rd:fp}, and $(iii)\Rightarrow(i)$ from normalization.
\end{proof}

Due to (i) of Lemma~\ref{lemlin}, we  can assume without loss of generality that the random variable $N$ satisfies the  condition $\rt(N)=-1$. Then condition (ii) of Lemma~\ref{lemlin} means that the conditional expectation of the ``num\'eraire'' $N$ is unique under all relevant probability measures, and condition (iii) can be viewed as \emph{additivity} with respect to the num\'eraire $N$:
\[
\rt(X+\lambda_tN)=\rt(X)-\lambda_t\quad \forall X\in\LTs,\quad \forall \lambda_t\in\Lts.
\]

The following proposition translates Lemma~\ref{lemlin} to the framework of risk measures for processes.
\begin{lemma}\label{proplin}
Let $\rt\,:\,\mathcal{R}_t^{\infty}\,\rightarrow\,\Lts$ be a conditional convex risk measure for processes, and let $N_s\in L^{\infty}_s$ for some $s\in\T_{t+1}$. Then the following conditions are equivalent:
\begin{itemize}
\item[(i)] $\rt(\lambda_tN_s1_{\T_s})=\lambda_t\rt(N_s1_{\T_s})$ for all $\lambda_t\in\Lts$;
\item[(ii)]$E_Q\left[-N_s\frac{D_s}{D_t}\,\Big|\,\F_t\,\right]=\rt(N_s1_{\T_s})$ for all $\bar Q=Q\otimes D\in\bar\Q_t^\alpha$ 
\item[(iii)] for all $X\in\mathcal{R}_t^{\infty}$ and $\lambda_t\in\Lts$
\[
\rt(X+\lambda_tN_s1_{\T_s})=\rt(X)+\lambda_t\rt(N_s1_{\T_s}).
\]
\end{itemize}
\end{lemma}

\begin{proof} 
Consider the conditional convex risk measure $\bar\rho_t\,:\,\LTp\,\rightarrow\,\Ltp$ associated to $\rt$  via \eqref{rrp}.
The linearity condition (i) for $\rt$ is equivalent to
\begin{equation*}%\label{linrb}
\bar\rho_t(\lambda_tN_s1_{\T_s})=\lambda_t\bar\rt(N_s1_{\T_s})\quad \forall \lambda_t\in\Lts,
\end{equation*}
i.e., $\bar\rho_t$ is linear on $\{\Lambda_t N_s1_{\T_s}\,|\,\Lambda_t\in\Ltp\}$. By Lemma~\ref{lemlin} and \eqref{rrp} this is equivalent to
\[
E_{\bar{Q}}[-N_s1_{\T_s}\,|\,\bar\F_t\,]=\rt(N_s1_{\T_s})1_{\T_t}\quad\bar Q\text{-a.s.}\quad \forall \bar{Q}=Q\otimes D\in\bar{\Q}_t^\alpha,
\]
and this is equivalent to (ii) by Corollary~\ref{cor:condexp}. In the same way, Lemma~\ref{lemlin} and \eqref{rrp} imply that  (i) is equivalent  to (iii).
\end{proof}

Since each  $D\in\mathcal{D}_t(Q)$ is non-decreasing, Lemma~\ref{proplin} applied to  $N_s=1$ for some $s>t$ yields the following characterization of cash additivity:
\begin{corollary}\label{proofca}
A conditional convex risk measure for processes $\rt\,:\,\mathcal{R}_t^{\infty}\,\rightarrow\,\Lts$ is cash additive at time $s\in\T_{t+1}$ if and only if
\[
D_t=D_{t+1}=\ldots=D_s\quad \qf
\]
for all $\bar Q=Q\otimes D\in\bar Q_t^\alpha$.
\end{corollary}

In other words,  cash additivity at time $s>t$ means that there is no discounting between $t$ and $s$ in all the relevant models. In particular we have the following proposition.

\begin{proposition}\label{propcash2}
A conditional convex risk measure for processes $\rt$ is cash additive at time $s\in\T_{t+1}$ if and only if it admits the robust representation 
\begin{equation}\label{ca}
\rt(X)=\es_{Q\in\Q^{\text{loc}}_t}\es_{\gamma\in\Gamma_{s}(Q)}\left(E_Q\lp-\sum_{k\in\T_s}\gamma_kX_k\mk\F_t\rp-\alpha_t(Q\otimes\gamma)\right),   \quad X\in\R_t^\infty.
\end{equation}
In this case $\rt$ is cash additive up to $s$, i.e., at all times $t+1\pk s$.

In particular, if $T<\infty$ or if $\T=\nat_0\cup\{\infty\}$, a risk measure for processes $\rt$ is  cash additive if and only if it reduces to a risk measure on $\LT$:
\begin{equation}\label{reprrv}
\rt(X)=\es_{Q\in\Q_t}\left(E_Q[-X_T|\F_t]-\beta_t(Q)\right),
\end{equation}
where $\beta_t(Q):=\alpha_t(Q\otimes\delta_{\{T\}})$, and $\delta_{\{T\}}$ denotes the Dirac measure at $T$.
\end{proposition}
\begin{proof}
Obviously representations \eqref{ca} implies cash additivity up to time $s$. The converse  follows from  1) of Corollary~\ref{cororep} and Corollary~\ref{proofca}. To prove the last part of the assertion, note that $\Gamma_T(Q)=\{\delta_{\{T\}}\}$ if $T<\infty$ or $\T=\nat_0\cup\{\infty\}$. Moreover, we have $Q\ll P$ for any $Q\in\Q_t^{\text{loc}}$ such that $Q\otimes\delta_{\{T\}}\in\bar Q_t^\alpha$. This is obvious for $T<\infty$, and it follows from Lemma~\ref{abscont} if  $\T=\nat_0\cup\{\infty\}$,  since $\gamma_\infty=1$ $Q$-a.s.\ in this case. Thus the  representation \eqref{reprrv} follows from \eqref{ca}.
\end{proof}

\begin{remark}
In particular, in the cash additive case and for $T<\infty$ or $\T=\nat_0\cup\{\infty\}$, the results of Section~\ref{tc} reduce to the corresponding results for risk measures for random variables from \cite{fp6, ap9}.
\end{remark}

The following example extends \cite[Proposition 2.4]{er08} to our present framework.

\begin{example}\label{exzcb}
Let $\rt\,:\,\mathcal{R}_t^{\infty}\,\rightarrow\,\Lts$ be a conditional convex risk measure for processes. Assume that there is a money market account $(B_t)_{t\in\T\cap\nat_0}$ as in Example~\ref{cashflow}, and that zero coupon bonds for all maturities $k>t, k\in\T\cap\nat_0$ are available at prices $B_{t,k}$, respectively.

Suppose that $\rt$ satisfies the following \emph{calibration condition}: 
\begin{equation}\label{Hansw}
\rt(\lambda_t\frac{B_{t}}{B_k}1_{\T_k})=-\lambda_tB_{t,k}\quad \forall \lambda_t\in\Lts,\quad\forall  k\in\T_t\cap\nat_0.
\end{equation}
Lemma~\ref{proplin} applied to $N_k=\dfrac{B_t}{B_k}$ implies that the calibration condition \eqref{Hansw} is equivalent to
\begin{equation*}%\label{addit}
\rt\left(X+\lambda_t\frac{B_t}{B_k}1_{\T_k}\right)=\rt(X)-\lambda_tB_{t,k}\quad \forall X\in\mathcal{R}_t^{\infty},\quad \forall \lambda_t\in\Lts,\quad \forall k\in\T_t\cap\nat_0,
\end{equation*}
and also to
\begin{equation}\label{na}
E_Q\left[\frac{B_t}{B_k}\frac{D_k}{D_t}\,\Big|\,\F_t\,\right]=B_{t,k}\quad\forall k\in\T_t\cap\nat_0,\quad\forall \bar Q=Q\otimes D\in\bar \Q_t^\alpha.
\end{equation}
Using \eqref{na}, the robust representation from part 1 of Corollary~\ref{cororep}, and monotone convergence for $T=\infty$, it can be seen that the calibration condition \eqref{Hansw} is equivalent to the following one, that may seem stronger at first sight:
\begin{equation*}%\label{Hans}
\rt\left(\sum_{k=t+1}^{T}\lambda_k\frac{B_t}{B_k}1_{\T_k}\right)=-\sum_{k=t+1}^{T}\lambda_kB_{t,k}\quad \forall \lambda_k\in\Lts.
\end{equation*}

Moreover, if the short rate process $(r_t)$, and hence also the money market account $(B_s)_{s\in\T\cap\nat_0}$ is predictable, then \eqref{na} implies
\[
\frac{B_t}{B_{t+1}}\frac{D_{t+1}}{D_{t}}=B_{t,t+1},
\]
and thus $D_{t+1}=D_t$ for all $\bar Q=Q\otimes D\in\bar \Q_t^\alpha$, since $B_{t,t+1}=(1+r_{t+1})^{-1}$ by a standard no arbitrage argument. Hence $\rt$ is cash additive at time $t+1$ by Corollary~\ref{proofca}.
In particular,  if a dynamic convex risk measure $(\rt)$ is time consistent, and if each $\rt$ satisfies the calibration condition \eqref{Hansw} with a predictable money market account, then each $\rt$ is cash additive by Proposition~\ref{infca}. In view of Remark~\ref{rmkca}, a time consistent dynamic convex risk measure that is continuous from above cannot satisfy condition \eqref{Hansw} for all $t\in\T$ if $\T=\nat_0$.
\end{example}

\section{Examples}\label{examplesec}
In this section we illustrate our analysis by discussing some examples, in particular analogues to classical risk measures for random variables such as the entropic risk measure and Average Value at Risk. Another class of examples is obtained by separating  model and discounting ambiguity in the robust representations of Subsection~\ref{robsec}.

\subsection{Entropic risk measures}
In this section we introduce entropic risk measures for processes. As a first variant we simply take  the usual conditional entropic risk measure on product space, that is the map $\bar{\rho}_t:\,\LTp\,\rightarrow\,\Ltp$ defined by
\begin{equation*}%\label{pentr}
\bar{\rho}_t(X)=\frac{1}{R_t}\log E_{\bar{P}}\lp e^{-R_tX}\mk\bar{\F}_t\rp
\end{equation*}
with risk aversion parameter $R_t=(r_0\pk r_{t-1},r_t,r_t,\ldots)\in\Ltp$, where $r_s>0$ and $r_s^{-1}\in L^{\infty}_s$ for all $s=0\pk t$, and $e^{-R_tX}=(e^{-r_sX_s})_{s\in\T}$. 

For an optional probability measure $\nu=(\nu_s)_{s\in\T}$ on $\T$, we denote by $\nu^t$ the normalized restriction to $\T_t$, i.e.
\begin{eqnarray*}%\label{normg} 
\nu_s^t = \left\{ \begin{array}{cl} \dfrac{\nu_s}{\sum_{j\in\T_t}\nu_j}, & \textrm{on}\,\, \lk\sum_{j\in\T_t}\nu_j>0\rk,\\
0, & \textrm{otherwise}
\end{array} \right.
\end{eqnarray*}
for $s\in\T_t$.
\begin{proposition}
The conditional entropic risk measure for processes $\rt\,:\,\mathcal{R}_t^{\infty}\,\rightarrow\,\Lt$ associated to $\bar{\rho}_t$ via \eqref{rrp} takes the form
\begin{equation}\label{entrp}
\rt(X)=\rho^{P,r_t}_t\left(-\rho^{\mu(\omega),r_t(\omega)}_t\left(X_.(\omega)\right)\right).
\end{equation}
Here $\rho^{P,r_t}_t:\,\LT\,\rightarrow\,\Lt$ denotes the usual conditional entropic risk measure for random variables with risk aversion parameter $r_t$:
\[
\rho^{P,r_t}_t(Y)=\frac{1}{r_t}\log E_P\lp e^{-r_tY}\mk \F_t\rp,\quad Y\in\LT.
\]
On the other hand,  $\rho^{\nu,r}_t:\,\real_b^{\T}\,\rightarrow\,\real$ is the entropic risk measure ``with respect to time", defined on the set of sequences $\real_b^{\T}=\{x=(x_s)_{s\in\T} |x_s\in\real\;\forall\,s,\; \sup_{s\in\T}x_s<\infty\}$ by
\begin{equation*}%\label{entime}
\rho^{\nu,r}_t(x)=\frac{1}{r}\log\left(\sum_{s\in\T_t} e^{-rx_s}\nu^t_s\right)
\end{equation*}
for a given probability measure $\nu$ on $\T$ and a risk aversion parameter $r\in\real$, $r>0$.

The minimal penalty function $\alpha_t$ of $\rt$ is given for $Q\otimes\gamma\in\map$ by
\begin{equation}\label{penen}
\alpha_t(Q\otimes\gamma)=\frac{1}{r_t}E_Q\lp\sum_{s\in\T_t}\gamma_s^t\log\frac{M_s}{M_t}\mk\F_t \rp+\frac{1}{r_t}E_Q[H(\gamma^t(\cdot)|\mu^t(\cdot))|\F_t],
\end{equation}
where $H(\cdot|\cdot)$ is the usual relative entropy for probability measures on $\T_t$, $M_s=\frac{dQ}{dP}|_{\F_s}$, $s\in\T\cap\nat_0$, and $M_\infty=\lim_{t\to\infty}M_t$ $P$-a.s. if $\T=\nat_0\cup\{\infty\}$.
\end{proposition}
\begin{proof}
Using Corollary~\ref{cor:condexp} we obtain
\begin{align*}
\bar{\rho}_t(X)&=-X1_{\{0\pk t-1\}}+\frac{1}{r_t}\log E\lp\sum_{s\in\T_t} e^{-r_tX_s}\mu^t_s\mk \F_t\rp1_{\T_t}\\
&=-X1_{\{0\pk t-1\}}+\rho^{P,r_t}_t\left(-\frac{1}{r_t}\log\left(\sum_{s\in\T_t} e^{-r_tX_s}\mu^t_s\right)\right)1_{\T_t}\\
&=-X1_{\{0\pk t-1\}}+\rho^{P,r_t}_t\left(-\rho^{\mu(\omega),r_t(\omega)}_t\left(X_.(\omega)\right)\right)1_{\T_t}.
\end{align*} 
To prove the second part of the claim, note that the minimal penalty function $\bar{\alpha}_t$ of $\bar{\rho}_t$ on $\map$ takes the form
\[
\bar{\alpha}_t(\bar{Q})=\frac{1}{R_t}H_t(\bar{Q}|\bar{P}),
\]
where $H_t(\bar{Q}|\bar{P})=E_{\bar{Q}}[\log \frac{{Z}_T}{{Z}_t}|\bar{\F}_t]$ is the conditional relative entropy of $\bar{Q}$ with respect to  $\bar{P}$, and $Z_s$ denotes the density of $\bar{Q}$ with respect to  $\bar{P}$ on $\bar\F_s$; see, e.g., \cite[Proposition 4]{dt5}. Using Theorem \ref{propq}, \eqref{dens}, Corollary~\ref{cor:condexp}, and \eqref{cexpd} we obtain for each $\bar{Q}=Q\otimes\gamma\in\map$,
\begin{equation*}%\label{alenp}
\bar{\alpha}_t(Q\otimes\gamma)=\frac{1}{r_t}E_Q\lp\sum_{s\in\T_t}\gamma^t_s\log\left(\frac{\gamma^t_sM_s}{\mu_s^tM_t}\right)\mk\F_t \rp1_{\T_t}.
\end{equation*}
Hence the minimal penalty function $\alpha_t$ of $\rt$ on $\map$ is given by
\begin{align*}%\label{penen}
\alpha_t(Q\otimes\gamma)&=\frac{1}{r_t}E_Q\lp\sum_{s\in\T_t}\gamma_s^t\log \frac{M_s}{M_t}\mk\F_t \rp+\frac{1}{r_t}E_Q\lp\sum_{s\in\T_t}\gamma_s^t\log\left(\frac{\gamma_s^t}{\mu_s^t}\right)\mk\F_t \rp \nonumber\\
&=\frac{1}{r_t}E_Q\lp\sum_{s\in\T_t}\gamma_s^t\log\frac{M_s}{M_t}\mk\F_t \rp+\frac{1}{r_t}E_Q[H(\gamma^t(\cdot)|\mu^t(\cdot))|\F_t].
\end{align*}
\end{proof}

One can characterize time consistency properties of the dynamic entropic risk measure for processes $(\rt)\zt$, where each $\rt$ is given by \eqref{entrp}, using the corresponding results for $(\bar\rt)_{t\in\T\cap\nat_0}$. In particular, by \cite[Proposition 37]{ap9} (cf.\ also \cite[Proposition 4.1.4]{ipen7}), the entropic risk measure \rtf is time consistent if the risk aversion parameter is constant, i.e., $r_t=r_0$ for all $t\in\T\cap\nat_0$, and \rtf is rejection (resp.\ acceptance) consistent if $r_t\ge r_{t+1}$ (resp.\ $r_t\le r_{t+1}$) for all $t\in\T\cap\nat_0$. 

\begin{remark}
A time consistent dynamic entropic risk measure \rtf is asymptotically precise under the reference measure $\bar P$, and hence under each $\bar{Q}\in\map$, due to Proposition~\ref{asymptoticpre}. Indeed, for each $X\in\Ri$ the supremum in the robust representation \eqref{newreprinf} of $\rho_0(X)$ is attained by a ``worst case'' measure $\bar Q^X\approx\bar P$ for each $X\in\Ri$; cf., e.g., \cite[Example 4.33]{fs4}.
\end{remark}

Formula~\eqref{penen} for the entropic penalty suggests to introduce a simplified version of the entropic risk measure, where the interaction between $Q$ and $\gamma$ in the penalty is reduced as follows:
For $u_t,v_t>0$ such that $u_t,v_t,u_t^{-1},v_t^{-1}\in\Lts$, define
\begin{equation}\label{penent}
\hat{\alpha}_t(Q\otimes\gamma):= \left\{ \begin{array}{c@{\quad\quad}l} \dfrac{1}{u_t}H_t(Q|P)+\dfrac{1}{v_t}E_Q[H(\gamma(\cdot)|\mu^t(\cdot))|\F_t], & \textrm{if}\,\,Q\in\Q_t,\gamma\in\Gamma_t(P),\\
\infty, & \textrm{otherwise}.
\end{array} \right.
\end{equation}
This induces a new conditional convex risk measure $\hat{\rho}_t:\,\mathcal{R}_t^{\infty}\,\rightarrow\,\Lt$ via
\begin{align*}%\label{enent}
\hat{\rho}_t(X)&:=\es_{Q\in\Q_t,\gamma\in\Gamma_t(P)}\left(E_Q\lp-\sum_{s\in\T_t}\gamma_sX_s\mk\F_t\rp-\hat{\alpha}_t(Q\otimes\gamma)\right)\\
&=\es_{\gamma\in\Gamma_t(P)}\rho_t^{P,u_t}\left(\sum_{s\in\T_t}\gamma_sX_s+\frac{1}{v_t}H(\gamma(\cdot)|\mu^t(\cdot))\right).
\end{align*}
\begin{proposition}
The conditional convex risk measure $\hat{\rho}_t$ satisfies
\begin{equation}\label{e}
\hat{\rho}_t(X)\leq\rho^{P,u_t}_t\left(-\rho^{\mu(\omega),v_t(\omega)}_t(X_.(\omega))\right).
\end{equation}
In particular, for $u_t=v_t=r_t$ we have 
\begin{equation}\label{i}
\hat{\rho}_t(X)\leq\rt(X)\quad\text{for all $X\in\mathcal{R}_t^{\infty}$},
\end{equation}
i.e., $\hat{\rho}_t$ is less conservative than the entropic risk measure $\rt$ in \eqref{entrp}.
\end{proposition}
\begin{proof}
Inequality \eqref{e} holds since 
 \[
\rho^{\nu,v_t}_t(x)=\sup\left\{-\sum_{s\in\T_t}y_sx_s-\frac{1}{v_t}H(y|\nu^t)\mk y=(y_s)_{s\in\T_t}\;\text{probability measure on}\;\T_t\right\}
\]
for any probability measure $\nu$ on $\T$.
\end{proof}
\begin{remark}
Inequality \eqref{i} implies the converse relation for the respective minimal penalty functions of $\hat{\rho}_t$ and $\rt$, and thus \eqref{penent} and \eqref{penen} yield
\[
H_t(Q|P)\geq E_Q\lp\sum_{s\in\T_t}\gamma_s\log M_s\mk\F_t \rp
\]
for all $Q\in\Q_t$ and $\gamma\in\Gamma_t(P)$.
\end{remark}

\subsection{Average Value at Risk}
For a given level $\Lambda_t=(\lambda_0\pk \lambda_{t-1},\lambda_t,\lambda_t,\ldots)\in\Ltp$ such that $\lambda_s\in(0,1]$ for all $s=0\pk t$ we define the conditional Average Value at Risk   $\bar{\rho}_t:\,\LTp\,\rightarrow\,\Ltp$ on the product space in the usual way as
\begin{equation*}%\label{avarp}
\bar{\rt}(X)=\es\{E_{\bar{Q}}[-X|\bar{\F}_t]\mk \bar{Q}\in\bar{\Q}_t, d\bar{Q}/d\bar{P}\leq \Lambda_t^{-1}\}.
\end{equation*}
\begin{proposition}
The conditional coherent risk measure for processes associated to $\bar{\rho}_t$ via \eqref{rrp} depends only on $\lambda_t$, and is given by
\begin{equation}\label{avarp2}
\rt^{\lambda_t}(X)=\es\lk E_Q\lp-\sum_{s\in\T_t}X_s\gamma_s\mk\F_t\rp\mk Q\in\Q_t^{\text{loc}}, \gamma\in\Gamma_t(Q), \frac{\gamma_sM_s}{\mu^t_s}\leq\lambda_t^{-1},\;s\in\T_t\rk,
\end{equation}
where $M_s=\frac{dQ}{dP}|_{\F_s}$, $s\in\T\cap\nat_0$, and $M_\infty=\lim_{t\to\infty}M_t$ $P$-a.s.\ if $\T=\nat_0\cup\{\infty\}$.
\end{proposition}
\begin{proof}
This is an immediate consequence of \eqref{dens} and Corollary~\ref{cor:condexp}.
\end{proof}

Note that a probability measure $Q$ and an optional measure $\gamma$ in the robust representation of $\rt$ are penalized simultaneously. As a simpler alternative, we can consider a ``decoupled" version of conditional Average Value at Risk, defined by
\[
\rt^{\lambda_1,\lambda_2}(X):=\es_{\gamma\in\Gamma_t^{\lambda_1}}AV@R_t^{\lambda_2}\left(\sum_{s\in\T_t}X_s\gamma_s\right),\quad X\in\mathcal{R}_t^{\infty}.
\]
Here, $\lambda_1$ and $\lambda_2$ are $\F_t$-measurable random variables with values in $(0,1]$,
\[
AV@R_t^{\lambda_2}(X)=\es\left\{E_Q[-X|\F_t]\mk Q\in\Q_t, \frac{dQ}{dP}\leq \frac{1}{\lambda_2}\right\}, \quad X\in\LT,
\]
is the usual Average Value at Risk for random variables, and
\[
\Gamma_t^{\lambda_1}=\left\{\gamma\in\Gamma_t(P) \mk \frac{\gamma_s}{\mu_s}\leq \frac{1}{\lambda_1},\; s\in\T_t\right\}.
\]
Note that $\rt^{\lambda_1,\lambda_2}$ is an example of a ``decoupled'' risk measure of the form \eqref{sepcoh}, which will be discussed in Subsection \ref{sep}.
\begin{proposition}
The conditional coherent risk measure $\rt^{\lambda_1,\lambda_2}$ satisfies
\[
\rt^{\lambda_1,\lambda_2}(X)\leq \rho_t^{\lambda_1\lambda_2}(X)\quad \forall\; X\in\R^\infty.
\]
In other words, the decoupled version is less conservative than the  conditional Average Value at Risk  defined in \eqref{avarp2} with $\lambda_t=\lambda_1\lambda_2$.
\end{proposition}
\begin{proof}
 Follows immediately from the definition of $\rt^{\lambda_1,\lambda_2}$.
\end{proof}

\begin{remark}
Recall that the dynamic Average Value at Risk for random variables is not time consistent; cf.\,e.g.\ \cite{adehk7}. Thus neither the dynamic Average Value at Risk for processes $(\rt^{\lambda_t})\zt$ defined in \eqref{avarp2}, nor its decoupled version $(\rt^{\lambda_1, \lambda_2})\zt$ will be time consistent in general. However, if the time horizon is finite, backward recursive construction of time consistent dynamic risk measures introduced in \cite[Section 4.2]{cdk6} (see also \cite[Sections 3.1, 4.1]{ck6}, \cite[Section 4.4]{ap9}) can be applied in order to obtain time consistent versions of  Average Value at Risk for processes and of its decoupled version. This can be done either on the product space using the construction from \cite[Sections 3.1]{ck6} or directly for risk measures for processes as in \cite[Sections 4.1]{ck6}. Indeed, it can be easily seen that if $(\bar{\rho}_t)\zt$ and $(\rho_t)\zt$ are associated to each other via \eqref{rrp}, the corresponding time consistent dynamic risk measures obtained by recursive construction will be also associated to each other via \eqref{rrp}.
\end{remark}

\subsection{Separation of model and discounting uncertainty}\label{sep}
If the time horizon $T$ is finite, we can replace  $\Gamma_t(Q)$ by $\Gamma_t(P)$ due to Remark~\ref{finite}, and the robust representation~\eqref{newreprinf} in Theorem~\ref{robdarpr} 
can be rewritten in the following form:
\[
\rt(X)=\es_{\gamma\in \Gamma_t(P)}\psi^{\gamma}_t\Big(\sum_{s=t}^TX_s\gamma_s\Big),\quad X\in\R^\infty_t.
\]
Here
\[
\psi^{\gamma}_t(Y)=\es_{Q\in\Q_t}\left(E_Q[-Y|\F_t]-\alpha_t(Q\otimes\gamma)\right),\quad Y\in \LT
\]
is a conditional convex risk measure for random variables (see, e.g., \cite[Theorem 1]{dt5}), that depends on the discounting factor $\gamma$ through its penalty function $\beta_t^\gamma(Q):=\alpha_t(Q\otimes\gamma)$. This formulation suggests a procedure to construct a simple class of conditional convex risk measures for processes, both for $T<\infty$ and $T=\infty$, where the dependence of $Q$ and $\gamma$ is separated in the following manner: One begins with some conditional convex risk measure for random variables $\psi_t\,:\,\LT\,\rightarrow\,\Lt$, specifies some set of discounting factors  $G_t\subseteq\Gamma_t(P)$, and defines 
\begin{equation}\label{sepcoh}
\rt(X)=\es_{\gamma\in G_t} \psi_t\left(\sum_{s\in\T_t}X_s\gamma_s\right),\quad X\in\R^{\infty}.
\end{equation}
It is easy to see that \eqref{sepcoh} defines a conditional convex risk measure $\rt$ for processes, and that $\rt$ is continuous from above if and only if $\psi_t$ is  continuous from above.

For example, for $G_t=\{\delta_{\{s\}}\}$ for some $s\in\T_t$, formula \eqref{sepcoh} reduces to
\[
\rt(X)=\psi_t(X_s),\quad X\in\R^{\infty},
\]
i.e., $\rt$ is a conditional convex risk measure on $L^{\infty}_s$. More generally, one can fix, as in \cite[Example 4.3.2]{ck6}, 
an optional measure $\gamma\in\Gamma_t(P)$, and define $G_t=\{\gamma\}$. In this case there is no ambiguity regarding the discounting process. For $T<\infty$ and $X\in\R^{\infty}_t$, or for $\T=\nat_0\cup\{\infty\}$ and $X\in\X^{\infty}_t$, we can switch to discounted terms by associating to $X$ a process $Y$ defined via
\[
Y_0:=X_0,\quad  \Delta Y_s:=D_s\Delta X_s,\quad s\in\T\cap\nat_0,\quad \textrm{and}\quad Y_\infty:=\lim_{t\to\infty}Y_t\quad \textrm{for}\quad\T=\nat_0\cup\{\infty\},
\]
where $D$ is related to $\gamma$ via \eqref{gammatodisc}. Then the risk measure $\rt$ defined by \eqref{sepcoh} reduces to a risk measure for random variables:
\begin{equation*}%\label{singledisc}
\rt(X)=\psi_t\Big(\sum_{s=t}^TD_s\Delta X_s\Big)=\psi_t\Big(\sum_{s=t}^T\Delta Y_s\Big)=\psi_t(Y_T).
\end{equation*}

A further example of a risk measure of the form \eqref{sepcoh} is given in \cite[Example 4.3.3]{ck6}; cf.\ also \cite[Example 4.2]{jr8}. Here we take
$G_t=\lk (1_{\{\tau=s\}})_{s\in\T_t}\mk \tau\in\Theta_t\rk$, where $\Theta_t$ denotes the set of all stopping times with values in $\T_t$. In this case 
\[
\rt(X)=\es_{\tau\in\Theta_t}\psi_t(X_{\tau}),
\]
is the maximal risk which arises by stopping the process  $(\psi_t(X_s))_{s\in\T_t}$ in the least favorable way.

\appendix
\section{Robust representations of risk measures for random variables}
The following definition of a conditional convex risk measure for random variables was given in \cite{dt5}: 

\begin{definition}\label{defrm}
A map $\rt\,:\,\LTs\,\rightarrow\,\Lt$ is called a \emph{conditional convex risk measure for random variables} if it satisfies the 
following properties for all $X,Y\in\LTs$:
\begin{itemize}
\item[(i)]
Conditional cash invariance: For all $m_t\in\Lt$,
\[\rt(X+m_t)=\rt(X)-m_t\]
\item[(ii)]
Monotonicity: $X\le Y\;\,\Rightarrow\;\,\rt(X)\ge\rt(Y) $
\item[(iii)]
Conditional convexity: For all $\lambda\in\Lt$ with $0\le \lambda\le 1$,
\[
\rt(\lambda X+(1-\lambda)Y)\le\lambda\rt(X)+(1-\lambda)\rt(Y)
\]
\item[(iv)]
{Normalization}: $\rt(0)=0$.
\end{itemize}
\end{definition}

The following theorem summarizes some robust representation results from  \cite[Theorem 1]{dt5}, \cite[Corollary 2.4]{fp6}, and \cite[Corollary 7]{ap9} that are often used in this paper.

\begin{theorem}\label{robdar}
Let $\rho_t\,:\,\LTs\,\rightarrow\,\Lt$ be a conditional convex risk measures for random variables. Then the following properties are equivalent:
\begin{enumerate}
\item
$\rho_t$ is continuous from above, i.e. 
\[
X^n\searrow X\;\,{P}\text{-a.s}\quad\Longrightarrow\quad \rho_t(X^n)\nearrow\rho_t(X)\;\,{P}\text{-a.s}
\]
for any sequence $(X^n)_n\subseteq\LTs$ and $X\in L^\infty$;
\item
for all ${Q}\in\ma$, $\rt$ has the robust representation
\begin{equation}\label{rd:ap}
\rho_t(X)=\qes_{{R}\in\Q_t({Q})}(E_{{R}}[-X\,|\,\F_t\,]-\alpha_t({R}))\quad {Q}\text{-a.s.},\qquad X\in L^\infty,
\end{equation}
where 
\[
 {\alpha}_t(Q)=\qes_{X\in\LTs}\left(E_Q[-X|\F_t]-\rt(X)\right)
\]
and
\begin{equation*}%\label{qtb}
\Q_t({Q}):=\lk {R}\in\ma\mk {R}={Q}|_{\F_t}\rk;
\end{equation*}
\item
$\rho_t$ has the robust representation
\begin{equation}\label{rd:dt}
\rho_t(X)=\es_{{Q}\in\Q_t}(E_{{Q}}[-X\,|\,\F_t\,]-\alpha_t({Q}))\quad {P}\text{-a.s.},\qquad X\in L^\infty,
\end{equation}
with $\Q_t:=\Q_t({P})$;
\item
$\rho_t$ has the robust representation
\begin{equation}\label{rd:fp}
\rho_t(X)=\es_{{Q}\in\Q^f_t}(E_{{Q}}[-X\,|\,\F_t\,]-\alpha_t({Q}))\quad {P}\text{-a.s.},\qquad X\in\ L^\infty,
\end{equation}
with
\[
\Q^f_t:=\lk {Q}\in\Q_t\mk E_{{Q}}[{\alpha}_t({Q})]<\infty\rk.
\]
\end{enumerate}
\end{theorem}

\section{It\^o-Watanabe decomposition}

The following is the discrete time version of the It\^o-Watanabe factorization of a nonnegative supermartingale; cf.\ \cite{iw65}.

\begin{proposition}\label{lemiw}
Let $U=(U_t)\zt$ be a nonnegative $P$-supermartingale on $(\Omega, \F, (\F_t)\zt, P)$ with $U_0=1$. Then there exist a nonnegative $P$-martingale $M=(M_t)\zt$ and a predictable non-increasing process $D=(D_t)\zt$ such that $M_0=D_0=1$ and
\begin{equation}\label{iwdec}
U_t=M_tD_t,\qquad t\in\T\cap\nat_0.
\end{equation}
Moreover such a decomposition is unique on $\{t<\tau_0\}$, where $\tau_0:=\inf\{t>0\,|\,U_t=0\}$.
\end{proposition}
\begin{proof}
We first assume that there exists a decomposition of $U$ as in \eqref{iwdec} and prove its uniqueness on $\{t<\tau_0\}$. Indeed, on $\{t<\tau_0\}=\{U_t>0\}=\{M_t>0\}\cap\{D_t>0\}$ we have
\[
\frac{E_P[U_{t+1}|\F_t]}{U_t}=\frac{D_{t+1}}{D_t}\frac{E_P[M_{t+1}|\F_t]}{M_t}=\frac{D_{t+1}}{D_t},
\]
and hence the process $D$ in the decomposition \eqref{iwdec} is uniquely determined on $\{t\leq \tau_0\}$ by 
\[
D_t=\prod_{s=0}^{t-1}\frac{E_P[U_{s+1}|\F_s]}{U_s},\quad 0\le t\leq \tau_0.
\]
Moreover, on $\{D_t>0\}$,
\[
M_t=\frac{U_t}{D_t},
\]
and thus also the process $M$ in the decomposition \eqref{iwdec} is uniquely determined on $\{D_t>0\}\supseteq\{t<\tau_0\}$.

To prove the existence of a decomposition as in \eqref{iwdec}, define the processes $D$ and $M$ via
\begin{equation*}
D_t = \left\{ \begin{array}{ll} \displaystyle{\prod_{s=0}^{t-1}\dfrac{E_P[U_{s+1}|\F_s]}{U_s}}, & \textrm{for}\; 0\leq t\leq \tau_0,\\
0, & \textrm{otherwise}
\end{array} \right.
\end{equation*}
and
\begin{equation*}
M_t = \left\{ \begin{array}{ll} \dfrac{U_t}{D_t}, & \textrm{on}\; \{D_t>0\},\\
M_{t-1}, & \textrm{on}\; \{D_t=0\}.
\end{array} \right.
\end{equation*}
Clearly, $D$ is predictable and non-increasing with $D_0=1$ and $D_t\geq 0$ for all $t$, and $M$ is adapted with $M_0=1$ and $M_t\geq 0$ for all $t$. It remains to show that $M$ is a martingale. Indeed,
\begin{align*}
E_P[M_{t+1}|\F_t]&=E_P[M_{t+1}1_{\{D_{t+1}=0\}}|\F_t]+E_P[M_{t+1}1_{\{D_{t+1}>0\}}|\F_t]\\
&=M_t1_{\{D_{t+1}=0\}}+\frac{1}{D_{t+1}}E_P[M_{t+1}D_{t+1}|\F_t]1_{\{D_{t+1}>0\}}\\
&=M_t1_{\{D_{t+1}=0\}}+\frac{1}{D_{t+1}}\frac{E_P[U_{t+1}|\F_t]}{U_t}U_t1_{\{D_{t+1}>0\}}\\
&=M_t1_{\{D_{t+1}=0\}}+\frac{1}{D_{t+1}}\frac{D_{t+1}}{D_t}U_t1_{\{D_{t+1}>0\}}\\
&=M_t,
\end{align*}
where we have used that $U_t>0$ on $\{D_{t+1}>0\}$.
\end{proof}

\begin{remark}
 Since $U$ is a nonnegative supermartingale, the following equivalence holds on $\{\tau_0=t\}$:
\[
D_t=0\, \Longleftrightarrow\, E_P[U_t|\F_{t-1}]=0\, \Longleftrightarrow\, P[U_t=0|\F_{t-1}]=1.
\]
Thus $D_t=0$ on the event $\{\tau_0=t\}$ if this event is sure at time $t-1$. On the other hand, we have $M_t=0$ on $\{D_t>0\}\cap\{\tau_0=t\}=\{E_P[U_t|\F_{t-1}]>0, U_t=0\}=\{P[U_t=0|\F_{t-1}]<1, U_t=0\}$, i.e., $M$ is uniquely determined also at time $\tau_0$ if $\tau_0$ is is not predicted one step ahead.
\end{remark}

\section{Disintegration of measures on the optional $\sigma$-field}\label{proof}
In this section we prove Theorem~\ref{propq}. Recall that we use Assumption~\ref{assfil}.
It guarantees that any consistent sequence of probability measures $Q_t$ on $\F_t$, $t\in\T\cap\nat_0$, admits a unique extension to a probability measure on $\F_\infty=\sigma(\cup_{t\in\T\cap\nat_0}\F_t)$, cf.\ \cite[Theorem 4.1]{par67}. In particular, any martingale $(M_t)\zt$ with $M_0=1$ induces a unique probability measure $Q$ on $(\Omega, F)$ such that
\begin{equation}\label{d}
 M_t=\frac{dQ}{dP}\Big|_{\F_t},\qquad t\in\T.
\end{equation}

\begin{proof}[Proof of Theorem \ref{propq}] Let $\bar Q\in\map$ with the density $\dfrac{d\bar{Q}}{d\bar{P}}=:\bar{Z}=(Z_t)\zte$. We first prove \eqref{expqg}  for $\T=\nat_0\cup\{\infty\}$.
To this end, consider the supermartingale $U=(U_t)\zte$ defined by
\begin{equation}\label{u}
U_t:=E_P\lp\sum_{s\in\T_t}\mu_sZ_s | \F_t\rp\ge 0,\quad t\in\T.
\end{equation}
By Proposition~\ref{lemiw}, $U$ admits a decomposition
\begin{equation*}
U_t=M_t D_t,\quad t\in\nat_0,
\end{equation*}
where $M=(M_t)_{t\in\nat_0}$ is a nonnegative $P$-martingale with $M_0=1$, and $D=(D_t)_{t\in\nat_0}$ is a nonnegative predictable non-increasing process with $D_0=1$. The martingale $M$ induces a unique probability measure $Q$ on $(\Omega,\F_\infty)$ via \eqref{d}, with $Q\in\mal$. Let $M_\infty:=\lim_{t\to\infty}M_t$ $P$-a.s., $D_\infty:=\lim_{t\to\infty}D_t$ $P$- and $Q$-a.s., and note that $Z_\infty\mu_\infty=U_\infty=\lim_{t\to\infty}U_t=M_\infty D_\infty$ $P$-a.s.. We define the process $\gamma=(\gamma_t)_{t\in\T}$ via \eqref{disctogamma}. Then for $X\in\Ri$ with $X\geq 0$ we have by monotone convergence and \eqref{u}
\begin{align*}
E_{\bar{Q}}[X]&=E_P\lp\sum_{t\in\T}X_t\mu_tZ_t\rp%=\sum_{t\in\T}E_P\lp X_t\mu_tZ_t\rp
=\sum_{t=0}^\infty E_P\lp X_tE_P[U_t-U_{t+1}|\F_t]\rp+E_P[M_\infty D_\infty X_\infty]\\
&=\sum_{t=0}^\infty E_P\lp X_t(M_tD_t-M_{t+1}D_{t+1})\rp+E_P[M_\infty D_\infty X_\infty]
=\sum_{t=0}^\infty E_Q[X_t\gamma_t]+E_P[M_\infty D_\infty X_\infty]\\
&=E_Q\lp\sum_{t=0}^\infty X_t\gamma_t\rp+E_P[M_\infty D_\infty X_\infty].
\end{align*}
Using \eqref{leb} this takes the form
\begin{equation}\label{leb1}
E_{\bar{Q}}[X]=E_Q\lp\sum_{t=0}^\infty X_t\gamma_t\rp+E_Q[X_\infty\gamma_\infty]-E_Q[\gamma_\infty X_\infty I_{\{M_\infty=\infty\}}].
\end{equation}
Plugging $X=1$ into \eqref{leb1} yields
\begin{align*}
1&=E_{\bar{Q}}[1]=E_Q\lp\sum_{t=0}^\infty\gamma_t+\gamma_\infty\rp-E_Q[\gamma_\infty I_{\{M_\infty=\infty\}}]\\
&=1-E_Q[\gamma_\infty I_{\{M_\infty=\infty\}}].
\end{align*}
Thus $\gamma_\infty=0$ $Q$-a.s.\ on $\{M_\infty=\infty\}$, i.e., $\gamma\in\Gamma(Q)$, and \eqref{leb1} reduces to \eqref{expqg}. 

To prove \eqref{expqg}  for $\T=\nat_0$, note that every measure $\bar Q$ on $(\Omega\times\nat_0, \bar\F)$ can be extended to a measure $\tilde Q$ on $(\Omega\times(\nat_0\cup\{\infty\}), \bar\F)$ by setting $\tilde Q[\Omega\times\{\infty\}]=0$. Thus  \eqref{expqg} yields 
\[
 E_{\tilde Q}[X]=E_Q\lp\sum_{t=0}^\infty X_t\gamma_t\rp+E_Q[X_\infty\gamma_\infty]
\]
with some probability measure $Q\in\mal$ and some optional measure $\gamma$ such that $\sum_{t=0}^\infty\gamma_t+\gamma_\infty=1$ $Q$-a.s.. Moreover, since 
\[
E_Q[\gamma_\infty]=E_{\tilde Q}[I_{\{\infty\}}]=0,
\]
we have $\gamma_\infty=0$ $Q$-a.s., i.e., $\gamma\in\Gamma(Q)$ for $\T=\nat_0$, and \eqref{expqg} holds.

Similarly, we can embed the case $\T=\{0\pk T\}$ into the setting of $\T=\nat_0\cup\{\infty\}$, by setting $\F_t:=\F_T$ for all $t>T$, and extending any measure $\bar Q$ on $(\Omega\times\{1\pk T\}, \bar\F)$ to a measure $\tilde Q$ on $(\Omega\times(\nat_0\cup\{\infty\}), \bar\F)$ by setting $\tilde Q[\Omega\times\T_{T+1}]=0$. The same reasoning as above yields a probability measure $Q\in\mal$, in particular $Q\ll P$ on $\F_T$, and an optional measure $\gamma$ such that $\gamma_s=0$ $Q$-a.s.\ for all $s>t$, i.e., $\gamma\in\Gamma(Q)$ for $\T=\{0\pk T\}$, and \eqref{expqg} holds.

The equality \eqref{expqd} follows from \eqref{expqg} due to integration by parts formula \eqref{intbyparts}.

To prove the converse implication of the theorem, note that each pair $(Q,\gamma)$, with $Q\in\mal$ and  $\gamma\in\Gamma(Q)$, defines a density $\bar Z=(Z_t)\zte$ of a probability measure $\bar Q\in\map$ via 
\begin{equation}\label{dens}
Z_t=\frac{M_t\gamma_t}{\mu_t},\qquad t\in\T,
\end{equation}
where $M_t$ denotes the density of $Q$ with respect to  $P$ on $\F_t$ for each  $t\in\T\cap\nat_0$, and, if $\T=\nat_0\cup\{\infty\}$, $M_\infty=\lim_{t\to\infty}M_t$ $P$-a.s.. Clearly, \eqref{expqg} and \eqref{expqd} hold for $\bar Q$.
\end{proof}

\begin{remark}\label{finite}
For $T<\infty$, and for $\T=\nat_0$, one can also prove Theorem~\ref{propq} directly, defining the supermartingale $U$ via \eqref{u} and using the It\^o-Watanabe decomposition of $U$ as above. For $T<\infty$, one obtains  in this way the additional property $\gamma\in\Gamma(P)$ in the decomposition $\bar Q=Q\otimes\gamma$ of any $\bar Q\in\map$, and so we can replace the set $\Gamma(Q)$ by $\Gamma(P)$ in the representation \eqref{newreprinf} and in all further results.
\end{remark}

\begin{remark}\label{rmku}
Let $\bar{Q}\in\map$ with decomposition $Q\otimes\gamma=Q\otimes D$ in the sense of \eqref{expqg} and \eqref{expqd}, let $\bar{Z}=(Z_t)\zte$ denote the density of $\bar{Q}$ with respect to $\bar{P}$, $M=(M_t)_{t\in\T\cap\nat_0}$ the density process of $Q$ with respect to $P$, and  $M_\infty=\lim_{t\to\infty}M_t$ $P$-a.s.\ for $\T=\nat_0\cup\{\infty\}$. 
\begin{enumerate}
\item The density $\bar Z$ takes the form \eqref{dens}. Indeed, for all $X\in\Ri$, $X\ge0$ we have
\[
E_{\bar{Q}}[X]=E_Q\lp\sum_{t\in\T}X_t\gamma_t\rp=E_P\lp\sum_{t\in\T}X_t\gamma_tM_t\rp,
\]
where, for $T=\infty$, the last equality holds due to monotone convergence, and, for $\T=\nat_0\cup\{\infty\}$, we use \eqref{leb} and $\gamma_\infty=0$ $Q$-a.s.\ on $\{M_\infty=\infty\}$. 
\item In order to clarify to which extent the decomposition \eqref{expqg} is unique, we note that the It\^o-Watanabe decomposition of the supermartingale $U$ defined in \eqref{u} is determined by the density process $M$ and the discounting process $D$. Indeed, 
\begin{align*}
U_t&=\sum_{s\in\T_t}E_P\lp \gamma_sM_s|\F_t\rp=\sum_{s\in\T_t}E_Q\lp \gamma_s|\F_t\rp M_t1_{\{M_t>0\}}\\
&=M_tE_Q\lp\sum_{s\in\T_t}\gamma_s\mk\F_t\rp1_{\{M_t>0\}}=M_tD_t,\qquad t\in\T,
\end{align*}
where we have used \eqref{dens}, \eqref{gammatodisc2}, and monotone convergence for $T=\infty$.
In particular, if $\bar{Q}\in\map$ admits two decompositions $\bar{Q}=Q^1\otimes D^1=Q^2\otimes D^2$, the uniqueness stated in Proposition~\ref{lemiw} yields
\[
M^1_t=M^2_t\quad \textrm{and}\quad D^1_t=D^2_t\quad \textrm{on}\;\; \{t<\tau_0\},
\]
where $\tau_0=\inf\{t>0\,|\,U_t=0\}$. Moreover, since $\bar{Z}>0$ $\bar{Q}$-a.s., we have $\bar{Q}[\{(\omega, t)| t\ge\tau_0(\omega)\}]=0$, and hence the processes $M$ and $D$ are uniquely determined and strictly positive $\bar{Q}$-a.s..
\item Equality $\bar Q$-almost surely between two processes $X,Y\in\Ri$ can be characterized as follows in terms of $Q$ and $\gamma$:
\begin{eqnarray*}%\label{qas}
X=Y\;\, \bar{Q}\textrm{-a.s.} &\Longleftrightarrow& 1=E_{\bar{Q}}[1_{\{X=Y\}}]=E_Q\lp \sum_{t\in\T}\gamma_t1_{\{X_t=Y_t\}}\rp\nonumber\\
&\Longleftrightarrow& X_t=Y_t\;\, \qf\;\, \textrm{on}\;\,\{\gamma_t>0\}\qquad \forall\, t\in\T,
\end{eqnarray*} 
where the last equivalence follows since $\sum_{t\in\T}\gamma_t=1\,Q$-a.s.. In particular, an $\bar\F_t$-measurable random variable $X=(X_t)\zte$ is well defined $\bar{Q}$-a.s.\ if and only if $X_i$ is well defined $Q$-a.s.\ on $\{\gamma_i>0\}$ for $i=0\pk t-1$, and $X_t$ is well defined $Q$-a.s.\ on $\{ \sum_{s\in\T_t}\gamma_s>0\}=\{D_t>0\}$. 
\end{enumerate}
\end{remark}

\begin{corollary}\label{cor:condexp}
For $\bar Q\in\map$ with decomposition $\bar Q=Q\otimes\gamma=Q\otimes D$, the conditional expectation given $\bar\F_t$ takes the form
\begin{equation*}%\label{condexp}
E_{\bar{Q}}[X\,|\,\bar\F_t\,]=X_01_{\{0\}}+\ldots+X_{t-1}1_{\{t-1\}}+E_Q\lp\sum_{s\in\T_t}\frac{\gamma_s}{D_t}X_s\mk\F_t\rp1_{\T_t},\qquad X\in\Ri,
\end{equation*}
where the last term on the right-hand-side is well defined $Q$-a.s.\ on $\{D_t>0\}$.
\end{corollary}
\begin{lemma}\label{abscont}
Let $\T=\nat_0\cup\{\infty\}$. For $\bar Q=\Q\otimes\gamma\in\map$ with the density process $(M_t)_{t\in\nat_0}$ of $Q$ with respect to $P$, and $M_\infty=\lim_{t\to\infty}M_t$ $P$-a.s.,  we have
 \[
  \gamma_\infty>0\quad Q\text{-a.s.}\quad\Leftrightarrow\quad Q\in\ma\quad\text{and}\quad \gamma_\infty>0\quad P\text{-a.s.}\;\,\text{on}\;\,\{M_\infty>0\}.
 \]
\end{lemma}
\begin{proof}
 We have
 \begin{align*}
 Q[\gamma_\infty>0]&=E_Q\lp\frac{\gamma_\infty}{\gamma_\infty}I_{\{\gamma_\infty>0\}}\rp=E_{\bar Q}\lp \frac{1}{\gamma_\infty}I_{\{\gamma_\infty>0\}}I_{\{\infty\}}\rp=E_{\bar P}\lp \frac{Z_\infty}{\gamma_\infty}I_{\{\gamma_\infty>0\}}I_{\{\infty\}}\rp\\
 &=E_P\lp \frac{Z_\infty\mu_\infty}{\gamma_\infty}I_{\{\gamma_\infty>0\}}\rp=E_P\lp M_\infty I_{\{\gamma_\infty>0\}}\rp,
 \end{align*}
where we have used \eqref{expqg} and \eqref{dens}. The claim follows by noting that $Q\ll P$ if and only if $E_P[M_\infty]=1$  for $Q\in\mal$ due to \eqref{leb}.
\end{proof}

Our robust representation of a conditional convex risk measure $\rt$ involves probability measures $\bar Q=Q\otimes\gamma$ which coincide on the $\sigma$-field $\bar\F_t$. This can be characterized as follows in terms of $Q$ and $\gamma$.
\begin{lemma}\label{rmkcon}
Let $\bar{Q}^1,\bar{Q}^2\in\map$ with the decompositions  $\bar Q^i=Q^i\otimes \gamma^i=Q^i\otimes D^i$, $i=1,2$. Then the following relation holds for all $t\in\T$
\[
\bar{Q}^1=\bar{Q}^2\;\,\text{on}\;\,{\bar\F_t}\quad \Longleftrightarrow\quad Q^1=Q^2\;\,\text{on}\;\,{\F_t}\cap\{D^1_t>0\}\;\, \textrm{and}\;\, \gamma^1_s=\gamma^2_s\;\, Q^1\text{-a.s.}\;\,\forall\, s<t.
\]
\end{lemma}
\begin{proof}
We denote by $\bar Z^i=(Z^i_t)\zte$ the density of $\bar Q^i$ with respect to  $\bar P$, by $(M^i_t)\zt$ the density process of $Q^i$ with respect to  $P$, and $M^i_\infty=\lim_{t\to\infty}M_t^i$ $P$-a.s.\ if $\T=\nat_0\cup\{\infty\}$, $i=1,2$. Assume that $\bar{Q}^1=\bar{Q}^2$ on ${\bar\F_t}$ for some $t\in\T$, i.e., $E_{\bar P}[\bar{Z}^1|\bar{\F}_t]=E_{\bar P}[\bar{Z}^2|\bar{\F}_t]$, where 
\begin{eqnarray}\label{cexpd}
E_{\bar P}[\bar{Z^i}|\bar{\F}_t]&=&Z^i_01_{\{0\}}+\ldots+Z^i_{t-1}1_{\{t-1\}}+\frac{1}{\sum_{s\in\T_t}\mu_s}E_P\lp\sum_{s\in\T_t}Z^i_s\mu_s\mk\F_t\rp1_{\T_t} \nonumber\\
&=&\frac{\gamma_0^iM_0^i}{\mu_0}1_{\{0\}}+\ldots+\frac{\gamma_{t-1}^iM_{t-1}^i}{\mu_{t-1}}1_{\{t-1\}}+\frac{D_t^iM_t^i}{\sum_{s\in\T_t}\mu_s}1_{\T_t},\qquad i=1,2
\end{eqnarray}
by \eqref{dens} and Corollary~\ref{cor:condexp}. This implies
\begin{equation}\label{coneq}
M^1_s\gamma^1_s=M^2_s\gamma^2_s\quad \forall\; s<t\quad \textrm{and}\quad M^1_tD^1_t=M^2_tD^2_t.
\end{equation}
Hence for any $A\in\F_{t-1}$ we obtain
\begin{align*}
E_P\lp\sum_{s=0}^{t-1}\gamma^1_sM^1_s1_A\rp&=E_P\lp M^1_{t-1}\sum_{s=0}^{t-1}\gamma^1_s1_A\rp=E_P\lp M^1_{t-1}(1-D^1_t)1_A\rp\\
&=Q^1(A)-\bar{Q}^1(A\times\T_t)\\
&=Q^1(A)-\bar{Q}^2(A\times\T_t),
\end{align*}
where the last equality follows since $A\times\T_t\in\bar{\F}_t$ and $\bar{Q}^1=\bar{Q}^2$ on ${\bar\F_t}$. In the same way we get
\[
E\lp\sum_{s=0}^{t-1}\gamma^2_sM^2_s1_A\rp=Q^2(A)-\bar{Q}^2(A\times\T_t).
\]
Therefore $Q^1=Q^2$ on ${\F_{t-1}}$, and by (\ref{coneq}) $\gamma^1_s=\gamma^2_s\;Q^1$-a.s. for all $s<t$. In particular $D^1_t=D^2_t\;Q^1$- and $Q^2$-a.s., which in turn implies $Q^1=Q^2$ on ${\F_t}\cap\{D^1_t>0\}$ due to (\ref{coneq}).\\
 The proof of the inverse implication works in the same way.
\end{proof}

\bibliographystyle{plain}

\end{document}